\documentclass[12pt]{article}
\pdfoutput=1

\usepackage{amsmath,amssymb,amsthm,graphicx}
\usepackage{changepage}
\usepackage{enumitem}
\usepackage{algorithm}
\usepackage{algorithmic}
\usepackage{scalerel}
\usepackage{accents}
\usepackage{xfrac}
\usepackage{thm-restate}
\usepackage{subfig}
\usepackage{tikz}
\usepackage{pgfplots}
\pgfplotsset{compat=newest}
\usetikzlibrary{shapes, shapes.symbols, backgrounds, matrix, calc, arrows, math, arrows.meta, decorations.pathreplacing, decorations.markings, patterns, intersections}
\pgfdeclarelayer{bg}
\pgfsetlayers{bg, main}
\usepackage{natbib}
\usepackage{nicefrac}
\usepackage[hidelinks]{hyperref}
\usepackage[capitalize]{cleveref}

\usepackage[margin=1.25in]{geometry}
\usepackage{setspace}
\theoremstyle{plain}
\newtheorem{theorem}{Theorem}
\newtheorem{lemma}{Lemma}
\newtheorem{claim}{Claim}

\newtheorem{proposition}{Proposition}
\newtheorem{corollary}{Corollary}

\theoremstyle{plain}
\newtheorem{definition}{Definition}
\newtheorem{example}{Example}
\newtheorem{assumption}{Assumption} 

\theoremstyle{plain}

\DeclareMathOperator*{\argmax}{arg\,max}

\newcommand{\prob}[2][]{\text{\bf Pr}\ifthenelse{\not\equal{}{#1}}{_{#1}}{}\!\left[{\def\givenn{\middle|}#2}\right]}
\newcommand{\expect}[2][]{\mathbb{E}\ifthenelse{\not\equal{}{#1}}{_{#1}}{}\!\left[{\def\givenn{\middle|}#2}\right]}

\newcommand{\rbr}[1]{\left(\,#1\,\right)}
\newcommand{\sbr}[1]{\left[\,#1\,\right]}

\newcommand{\wel}{{\rm Wel}}

\newcommand{\dd}{\,{\rm d}}

\newcommand{\virtual}{\varphi}
\newcommand{\ironed}{\bar{\virtual}}
\newcommand{\opt}{{\rm OPT}}
\newcommand{\vwel}{{\rm VWel}}
\newcommand{\apx}{{\rm APX}}

\newcommand{\alloc}{x}
\newcommand{\util}{u}

\include{math}
\include{nicefrac}

\DeclareMathOperator\erf{erf}

\title{
Approximate Optimality of Linear Contracts\\Under Uncertainty%
\thanks{
A one-page abstract of this paper was accepted at EC 2023 under the title ``Bayesian Analysis of Linear Contracts". 
This work received funding from the European Research Council (ERC) under the European Union's Horizon 2020 research and innovation program (grant agreement No.~101077862), the Israel Science Foundation (Grant No.~336/18), and a Google Research Scholar Award. Part of the work of Y.~Li was done while he was a PhD candidate at Northwestern University under the support of NSF Grant SES-1947021, and a Postdoc at Yale under the support of Sloan Research Fellowship FG-2019-12378. 
We wish to thank Gabriel Carroll, Danil Dmitriev, Mira Frick, Ryota Iijima, Jonathan A. Libgober, Ilya Segal, Zhuoran Yang, anonymous reviewers of an earlier version of this work, and audiences at the ACM Conference on Economics and Computation and the Stony Brook Game Theory Conference for their feedback.
}}

\author{Tal Alon\thanks{Department of Computer Science, The Technion -- Israel Institute of Technology, Haifa, Israel. Email: \texttt{alontal@campus.technion.ac.il}} \and 
Paul D\"utting\thanks{Google Research, Zurich, Switzerland. Email: \texttt{duetting@google.com}} \and 
Yingkai Li\thanks{Department of Economics, National University of Singapore, Singapore.
Email: \texttt{yk.li@nus.edu.sg}} \and 
Inbal Talgam-Cohen\footnote{Department of Computer Science, The Technion -- Israel Institute of Technology, Haifa, Israel. Email: \texttt{inbaltalgam@gmail.com}}}

\date{}

\begin{document}

\maketitle
\begin{abstract}
We consider a hidden-action principal-agent model, in which actions require different amounts of effort, and the agent privately knows his ability that determines his cost of effort. We show that linear contracts admit approximation guarantees that improve with a natural metric that captures the degree of uncertainty in the contracting setting. We thus show that linear contracts are near-optimal whenever there is enough uncertainty. In contrast, other simple contract formats such as debt contracts may suffer from a loss linear in the number of possible actions, even when there is sufficient uncertainty. 
\end{abstract}

\section{Introduction}
\label{sec:intro}



Contract design is a pillar of modern economic theory, serving as the cornerstone for establishing incentives between parties with moral hazard concerns. 
In the fundamental model of contract theory, the hidden-action principal-agent model \citep{Holmstrom79,GrossmanHart83}, the principal seeks to incentivize an agent to take a costly action. 
The principal does not directly observe the action taken by the agent, but different actions result in distinct reward distributions for the principal. This allows the principal to influence the agent's choice of action by designing a contract that specifies payment to the agent contingent on rewards.
The classic theory focuses on the design of optimal contracts that achieve the highest possible revenue (as given by the principal's expected reward minus payment to the agent under the action chosen by the agent). This requires sophisticated incentive structures tailored to specific contexts. 
However, despite their theoretical appeal, the practical implementation of these optimal contracts can be challenging, as it involves navigating trade-offs, conducting extensive analysis, allocating substantial computational resources, and communicating complex terms to agents.

In contrast, linear contracts, which involve a fixed share of output as payment to the agent, play a crucial role in practice due to their simplicity and versatility. 
They are widely adopted across different economic and business settings, offering straightforward terms that are easier to understand, communicate, and implement.
\citet{carol-robust-linear} shows that linear contracts are robustly optimal when the principal is uncertain about the agent's additional actions and aims to maximize worst-case performance. Similar results have also been obtained under alternative models of ambiguity. 
This robust justification helps explain the prevalence of linear contracts in practice. 
However, as initially conjectured by \citet{milgrom-holmstrom87}, such justification is unlikely to hold in traditional Bayesian models. 
Indeed, in the analysis of optimal contracts, linear contracts are shown to deliver suboptimal performance in various applications under 
Bayesian models \citep[e.g.,][]{dutting2020simple, CastiglioniMG22, guruganesh20}.


We analyze the performance of linear contracts in Bayesian settings, by employing the theory of approximation, which complements the theory of optimality and provides justifications for various mechanisms commonly observed in practice
\citep[e.g.,][]{hartline2012approximation,roughgarden2019approximately}. 
We apply this approach to a natural single-dimensional private types model. In this model, a principal interacts with an agent to select one of $n$ costly actions. 
Each action is associated with a specific effort requirement and a stochastic output distribution that determines the reward of the principal. 
The agent has a private type, drawn from publicly-known distribution, that represents his ability, and his cost of effort is lower for all actions if he has a higher ability.
The principal's goal is to design a (possibly randomized) menu of contracts that maximizes her expected revenue.

As shown by \citet{dutting2020simple}, linear contracts are not approximately optimal in the degenerate case of our model where the agent's type distribution is a point mass, i.e., there is no uncertainty regarding the agent's type.\footnote{Approximation guarantees for linear contracts are also given in \citep{BalamcedaEtAl16}, but for the objective of welfare rather than revenue.} 
Specifically, the approximation factor of linear contracts deteriorates linearly with the number of possible actions in the environment. However, in this paper, we reveal that this is a critical boundary case, and linear contracts become approximately optimal when there is sufficient uncertainty and the setting is therefore \emph{not}  ``pointmass-like.''
We do so by demonstrating that linear contracts admit approximation guarantees that improve with a natural metric of uncertainty.


Our definition of uncertainty (\cref{def:small tail value}) encompasses the entire principal-agent instance, rather than solely relying on the agent's type distribution. Specifically, we consider there to be sufficient uncertainty in the setting if there exists a quantile $q$ such that the optimal welfare obtained from types within a small neighborhood around the low-cost types,
with a measure of $q$, is bounded away from the optimal welfare obtained from all types.\footnote{Similar definitions for bounding the contributions of tail events can be found in the literature of sample complexity \citep{devanur2016sample}.} At first glance, this definition may appear unusual since one might argue that welfare should not be a factor in quantifying type uncertainty. However, we show that this dependence on the principal-agent instance is necessary and well justified. 

To illustrate this point, let us consider a contracting setting 
where the agent's type is drawn from a uniform distribution in the range $[1, 100]$ where higher types correspond to lower abilities. In this case, there is substantial uncertainty about the unknown type in the traditional sense, characterized by high entropy and large variance. 
However, if, for a small constant $\epsilon > 0$, the contract instance is such that only agents with types between $[1, 1+\epsilon]$ are incentivized to choose costly actions even when the full reward from the stochastic output is given to the agent, 
the contract design instance can effectively be reduced to the case where the principal knows that the agent's true type lies within $[1, 1+\epsilon]$. 
For this reason, 
the underlying uncertainty in this scenario is actually small, and the revenue gap between linear contracts and optimal contracts can be large.\footnote{This point is stated more formally in \cref{exp:thin_tail} of \cref{sec:linear}.} 

The proposed definition of uncertainty takes into account such corner cases, and we show that linear contracts are approximately optimal with approximation factors that depend on the level of uncertainty in the instance. This result is formally defined and quantified in \cref{thm:universal}. 
Importantly, under the condition of sufficient uncertainty, the revenue generated by linear contracts is comparable to the optimal welfare, which serves as an upper bound for the optimal revenue.

We provide two refined bounds on the approximations of optimal linear contracts by incorporating additional structures on the type distribution or the marginal cost function. Firstly, if the type distribution satisfies a slowly-increasing condition (\cref{def:slow-inc-dist}), meaning that the increase in the cumulative distribution function is slow, linear contracts can achieve better approximations to the optimal welfare. The intuition behind this is that the revenue generated by low-cost types in linear contracts can approximately cover the welfare contribution from high-cost types, and the welfare contribution from a small range of low-cost types is not significant compared to the optimal welfare due to the presence of sufficient uncertainty in the instance.

Secondly, all these results and conditions for the cost space can be naturally extended to the virtual cost space, with some adjustments, assuming additionally that the cost function is linear. The virtual cost, as interpreted by \citet{bulow1989simple}, represents the marginal cost to the agent of exerting additional effort. 
Our condition on the virtual cost space (\cref{def:lin-bounded-dist}) can be understood as the property that the marginal cost function has a bounded range and is sandwiched between two linear functions. %
We show that in this case linear contracts achieve approximate optimality in comparison to the optimal revenue rather than the optimal welfare.
This result enables us to achieve improved approximation factors for a broad class of instances, as the optimal revenue benchmark is less stringent than the optimal welfare benchmark.

We use our general theorems to show that linear contracts achieve constant-factor approximations to the optimal welfare and/or revenue in a wide range of settings (see Sections~\ref{sub:slow-inc} and \ref{sub:lin-bound}). 
Additionally, we show (in Theorem~\ref{thm:upper-bound-n}) that even without the assumption of sufficient uncertainty, the approximation factor of linear contracts to the optimal revenue scales at most linearly with the number of actions.
Therefore, when the number of available actions is small, linear contracts are approximately optimal.
Moreover, since linear contracts do not screen the agent even when there are private types, 
our approximation result implies that the value of screening is limited in those settings.

We complement our positive results for linear contracts by showing that for other formats of simple contracts such as debt contracts and single-outcome payment contracts, even when there is sufficient uncertainty regarding the agent's ability, the worst-case approximations of those simple contracts can be linear in the number of possible actions (\cref{prop:debt,prop:single}). 

Finally, the interpretation of our approximation results is not to take the exact approximation factors too literally.
Instead, the economics of these approximation results are derived through relative comparisons. In particular, by comparing the worst case approximation of linear contracts in environments with and without sufficient uncertainty, where the former allows constant approximations while the latter has approximations degrades linear in the number of actions, we show that sufficient uncertainty is the main driving force for the approximate optimality of linear contracts. In addition, by comparing the worst case approximation of linear contracts with other simple contracts such as debt contracts, we show that linear contracts significantly outperform other simple contracts in the worst case when there is sufficient uncertainty.


\subsection{Related Work}\label{sec:related-work} 

A classic reference for justifying linear contracts is \citep{milgrom-holmstrom87}, which establishes the optimality of linear contracts in a dynamic setting. 
In another pioneering work, \citet{Diamond98} shows that in a setting where the agent can either exert effort or not, and can freely choose the distribution over outcomes subject to a moment constraint, a linear contract for inducing effort is no more expensive than any other contract.
A number of recent studies have established the robust (max-min) optimality of linear contracts in non-Bayesian models \citep{carol-robust-linear,dutting2020simple,YuK20,DaiT22,WaltonC22}.
\citet{kambhampati2023randomization} show that in some cases the max-min performance can be improved by using randomized contracts, which alleviate the principal's ambiguity aversion. 
\citet{kambhampati2025randomization} show that robustly optimal randomized contracts are again linear.
In our paper, we consider a Bayesian model where linear contracts can be suboptimal. However, we show that a deterministic linear contract is approximately optimal even compared to the optimal welfare 
under the sufficient uncertainty condition. 
We thus provide an alternative justification for linear contracts: they are near-optimal provided there is sufficient Bayesian uncertainty.

Our work also contributes to the broad literature on problems that combine moral hazard with adverse selection \citep[see, e.g.,][]{Myerson82}. Two closely related papers are \citet{gottlieb2022simple} and \citet{castro2024disentangling}. \citeauthor{gottlieb2022simple} provide conditions under which it is optimal for the principal to offer a single debt contract. \citeauthor{castro2024disentangling}~examine when moral hazard and private types can be decoupled. So both these works identify conditions under which the optimal menu of contracts 
is simple. In contrast, we provide conditions under which a particularly relevant class of simple contracts, namely linear contracts, achieve near-optimal (rather than optimal) performance.

The worst-case approximation approach that we take here is gaining traction in economic theory  \citep[e.g.][]{hartline2015non,akbarpouralgorithmic}. 
This analysis framework has been adopted for the objective of maximizing the principal's revenue
in pure-moral hazard settings by \citet{dutting2020simple}, and subsequently by \citet{CastiglioniMG22} and \citet{guruganesh20} in Bayesian settings. 
\citet{dutting2020simple} give worst-case approximation guarantees in the natural parameters of the model, and in particular show that the worst-case gap between linear and optimal scales linearly in the number of actions.
\citet{CastiglioniMG22} study a multi-dimensional private type setting where both the cost of effort and the outcome distribution associated with each action are private information of the agent. 
They show that compared to optimal contracts, linear contracts may suffer from a multiplicative loss that is at least linear in the number of possible types, even in binary action settings. 
\citet{guruganesh20} consider a setting where only the outcome distribution is private and all types share the same costs, and show 
that the gap between linear contracts and optimal welfare 
scales at least linearly in the number of actions, and logarithmically in the number of types.
In contrast, in our model, the only uncertainty lies in the agent's skill level, which determines their cost per unit-of-effort. 
We show that with sufficient uncertainty, linear contracts are approximately optimal. Furthermore, the approximation factor is at most linear in the number of actions, independently of the number of types, even without such assumptions.

There is also a recent advancement in the computer science literature for computing  optimal contracts in the presence of adverse selection. 
In the case of single-dimensional private types, as considered in our paper, 
\citet{AlonDT21} show that the optimal deterministic menu of contracts can be computed in polynomial time for constant number of actions, although there may exist randomized menus that outperform it.
In the case of multi-dimensional private types, 
\citet{guruganesh20} show that computing the optimal deterministic menu of contracts is NP-hard, even for a constant number of actions,
while \citet{castiglioni2022designing} show that the computation of the optimal randomized menu contract can be reduced to a linear program and solved efficiently. 
Our paper takes an approach that is orthogonal to this line of work, by providing approximation guarantees for linear contracts.
Follow-up work by \citet{castiglioni2025reduction} 
gives polynomial-time algorithms for converting any approximately optimal contract in single-dimensional settings, as in our paper, to multi-dimensional settings where the agent also has private information about outcome distributions, while preserving the approximation guarantee.

Finally, there is a growing body of work that explores combinatorial contracts through a computational lens \citep[e.g.,][]{Babaioff2006combinatorial,DuttingEFK21,DuttingRT21,DuettingEFK23,DuttingEFK25}.

\section{Preliminaries} 
\label{sec:prelim}

Throughout, let $[n]=\{0,1,...,n\}$ 
(zero is included).\footnote{Our main results can be easily extended to environments with a continuum of actions. } 
We consider a single principal interacting with a single agent whose ability is unknown to the principal. 

\subsection{A Principal-Agent Model}
In a principal-agent instance, there is an \emph{action} set~$[n]$ from which the agent chooses an action. Each action $i\in[n]$ has an associated distribution $F_{i}$ over an \emph{outcome} set $[m]$. The $j$th outcome is identified with its \emph{reward} $r_j\geq 0$ to the principal. 
W.l.o.g.~we assume increasing rewards $r_0 \le \ldots \leq r_m$. We denote the probability for reward~$r_j$ given action $i$ by $F_{i,j}$, and the expected reward of action $i$ by $R_i = \sum_j F_{i,j} \cdot r_j$. Each action~$i$ requires $\gamma_i\geq 0$ \emph{units-of-effort} from the agent. We assume w.l.o.g.~that actions are ordered by the amount of effort they require, i.e., $\gamma_0 < \ldots < \gamma_n$. 
Moreover, we assume that an action which requires more units-of-effort has a strictly higher expected reward, i.e., every two actions $i<i'$ have distinct expected rewards $R_i < R_{i'}$ (the meaning of this assumption is that there are no ``dominated'' actions, which require more effort for less expected reward). 

We assume that action $i=0$, referred to as the \emph{null action}, requires no effort from the agent ($\gamma_0=0$). This action deterministically and uniquely yields the first outcome, referred to as the \emph{null outcome}, i.e., $F_{0,0} = 1$ and $F_{i,0} = 0$ for any~$i \geq  1$. This models the agent’s opportunity to opt-out of the contract if individual rationality (the guarantee of non-negative utility) is not satisfied.

\paragraph{Types and utilities.} 
The agent has a privately ability \emph{type $c \geq 0$}, also known as his \emph{cost}, as this parameter captures the agent's cost per unit-of-effort. 
We assume that the private type is drawn from a publicly-known distribution~$G$, with density $g$ and support $C=[\underline{c},\bar{c}]$ for $0\leq \underline{c}<\bar{c}\leq \infty$. 
The cost of exerting effort $\gamma$ for type $c$ is denoted as $\zeta(c, \gamma)$.\footnote{This is a generalization of the single-dimensional type model considered in \citet{AlonDT21} where $\zeta(c, \gamma) = c\cdot \gamma$. Their cost function satisfies \cref{asp:cost_structure} in our model.}

\begin{assumption}\label{asp:cost_structure}
The cost function $\zeta$ satisfies that $\zeta(0, \gamma)=0$ and $\zeta(c, \gamma)>0$ for all $\gamma> 0$ and $c > 0$.
Moreover, given any pair of costs $0 < c < c'$, 
\begin{enumerate}
    \item $\zeta(c,\gamma)$ is strictly increasing in $\gamma$ and $c$;
    \item $\frac{\zeta(c,\gamma)}{c}$ is weakly increasing in $c$ for any $\gamma> 0$;
    \item $\frac{\zeta(c',\gamma)}{\zeta(c,\gamma)}$ is weakly increasing in $\gamma$, i.e., $\zeta$ is log-supermodular.
\end{enumerate}
\end{assumption}

In \cref{asp:cost_structure}, the first condition is very natural. It implies that the cost of the agent is always higher if the effort level is higher or the cost type of the agent is higher (higher cost type corresponds to lower ability of the agent). 
The second condition is almost a tautology given that $\zeta(c,\gamma)$ is strictly increasing in $c$. 
This is because the actual cardinal value of type $c$ has no practical meaning at this point, and we can re-scale it to ensure that the second condition holds. We choose this particular scaling to simplify the expositions in later sections. 

The last condition in \cref{asp:cost_structure} is more salient. 
It is equivalent to the condition that for any $0 < c < c'$ and $\gamma < \gamma'$,
\begin{align}\label{eq:equiv_ineq}
\frac{\zeta(c,\gamma')}{\zeta(c,\gamma)}
\leq \frac{\zeta(c',\gamma')}{\zeta(c',\gamma)}.
\end{align}
It implies that low ability (high cost) agents will find it more costly (measured in terms multiplicative increase in costs) to choose higher effort actions compared to high ability (low cost) agents. 


We assume risk-neutrality and quasi-linear utilities for both parties as in \citet{carol-robust-linear}. 
Given payment $t$ from the principal to the agent,
the utilities of the agent and the principal for choosing action~$i$ with a realization of reward $r_j$ are 
\begin{align*}
u_A = t-\zeta(c,\gamma_i)
\quad\text{and}\quad
u_P = r_j-t
\end{align*}
respectively. 
We also refer to the principal's utility as her \emph{revenue}.

\subsection{Contracts}\label{sec:contracts}
We focus on deterministic contracts to simplify the exposition.\footnote{It is known from \citet{AlonDT21} and \citet{castiglioni2022designing} that considering \emph{randomized} allocation rules (together with randomized payments) can result in contracts that increase the principal's revenue. Our characterization of implementable contracts and our approximation guarantees
all hold for randomized contracts.}
A \emph{contract} $(x,t)$ is composed of an \emph{allocation rule} $x:C\to [n]$ that maps the agent's type to a recommended action,
and a \emph{payment (transfer) rule} $t: C\to \mathbb{R}_+^{m+1}$
that maps the agent's type to a profile of $m+1$ payments (for the $m+1$ different outcomes). 
We use $t^{\hat{c}}_j$ to denote the payment for outcome~$j$ given type report~${\hat{c}}$. 
We consider contracts where all transfers are required to be non-negative, i.e., $t^{\hat{c}}_j\geq 0$ for all $\hat{c}$ and~$j$. 
This guarantees the standard \emph{limited liability} property for the agent, who is never required to pay out-of-pocket \citep{innes1990limited,carol-robust-linear}.

The timeline of a contract $(x,t)$ is as follows:
\begin{itemize}
\item The principal commits to a contract $(x,t)$.
\item The agent with type $c$ submits a report $\hat{c}$, and receives an action recommendation~$x(\hat{c})$ and a payment profile $t^{\hat{c}}=(t^{\hat{c}}_0,...,t^{\hat{c}}_m)$. 
\item The agent privately chooses an action $i$, which need not be the recommended action~$x(\hat{c})$. 
\item An outcome $j$ is realized according to distribution $F_i$. 
The principal is rewarded with~$r_j$ and pays a price of $t^{\hat{c}}_j$ to the agent. 
\end{itemize}

A special class of contracts that we are interested in this paper is linear contracts.
Note that even when the agent has private types, screening is not possible under linear contracts because all types of the agent will prefer the largest possible share. 
\begin{definition}
A contract $(x,t)$ is \emph{linear with parameter $\alpha\geq 0$} if $t^c_j=\alpha r_j$ for any $c\in C$ and $j \in [m].$
\end{definition}


\paragraph{Incentives.} 
Let~$c$ be the true type, $\hat{c}$ be the reported type, $i$ be the action chosen by the agent, and $j$ be the realized outcome. 
Denote by $T^{\hat{c}}_i=\sum_{j\in [m]} F_{i,j} t^{\hat{c}}_j$ the expected payment for action $i$. 
In expectation over the random outcome $j\sim F_i$, 
the expected utilities of the agent and the principal are
\begin{align*}
T^{\hat{c}}_i-\zeta(c,\gamma_i)
\quad\text{and}\quad
R_i-T^{\hat{c}}_i
\end{align*}
respectively.
Notice that the sum of the players' expected utilities is the expected \emph{welfare} from action $i$ and type $c$, namely $R_i-\zeta(c,\gamma_i)$.

For any agent with true type $c$, 
when he faces a payment profile $t^{\hat{c}}$, he will choose an action 
\begin{equation}
i^*(t^{\hat{c}},c) \in \arg\max_i\, T^{\hat{c}}_i-\zeta(c,\gamma_i).\label{eq:i-star}
\end{equation}
As is standard in the contract design literature, if there are several actions with the same maximum expected utility for the agent, we assume consistent tie-breaking in favor of the principal \citep[see, e.g.,][]{guruganesh20}.
Thus $i^*(t^{\hat{c}},c)$ is well-defined.
Therefore, when reporting his type, the agent will report $\hat{c}$ that maximizes his expected utility given his anticipated choice of action, i.e., $T^{\hat{c}}_{i^*(t^{\hat{c}},c)}-\zeta(c,\gamma_{i^*(t^{\hat{c}},c)})$. 

\begin{definition}
A contract $(x,t)$ is \emph{incentive compatible (IC)} if for every type $c\in C$,
the expected utility of an agent with type $c$ is maximized by: 
\begin{enumerate}[label={(\roman*)}]
    \item following the recommended action, i.e., $x(c)=i^*(t^c,c)$;
    \item truthfully reporting his type, i.e., $c\in \arg\max_{\hat{c}\in C} \{T^{\hat{c}}_{i^*(t^{\hat{c}},c)}-\zeta(c,\gamma_{i^*(t^{\hat{c}},c)})\}$. 
\end{enumerate}
\end{definition}

\paragraph{The principal's objective.} The principal's goal (and our goal in this paper) is to design an \emph{optimal} IC contract~$(x,t)$, i.e., an incentive compatible contract that maximizes her expected revenue. 
The expectation is over the agent's random type $c$ drawn from $G$, as well as over the random outcome induced by the agent's prescribed action $x(c)$.
The principal maximizes $\mathbb{E}_{c\sim G}[R_{x(c)}-T^{c}_{x(c)}]$, or equivalently, $\mathbb{E}_{c \sim G,j \sim F_{x(c)}} [r_j - t^c_j]$, subject to the IC constraints. 
Thus, the optimal revenue is 
\begin{align*}
\opt = \max_{(x,t) \text{ is IC} }
\mathbb{E}_{c\sim G}[R_{x(c)}-T^{c}_{x(c)}].
\end{align*}
For any simple contract (e.g., linear contract) with expected revenue $\apx$, 
we say this contract is a $\beta$-approximation to the optimal revenue if $\beta\cdot\apx \geq \opt$. 


\section{Near-Optimality of Linear Contracts}
\label{sec:linear}


In this section, we establish our main results: approximation guarantees for arguably the simplest non-trivial class of contracts, namely linear contracts. 
We first identify a crucial condition, the thin-tail parameterization, 
which smoothly separates the knife-edge case of point-mass like settings, from typical settings with sufficient uncertainty on the agent's ability. 
In the former case, linear contracts can perform arbitrarily badly with respect to the optimal revenue \citep[c.f.,][]{dutting2020simple}, 
while in the latter case they can achieve a constant approximation. 
Let 
$$\wel_{[a,b]}\triangleq\int_{a}^{b} R_{i^*(r,c)}-\zeta(c,\gamma_{i^*(r,c)}) \dd G(c)$$
be the welfare contribution from types within $[a,b]$
where 
$$i^*(r,c) \in \argmax_{i} \expect[r\sim F_i]{r} - \zeta(c,\gamma_i)
=\argmax_{i} R_i - \zeta(c,\gamma_i)$$ 
is the optimal action that maximizes expected welfare for type $c$.

\begin{definition}[Thin-tail parameterization]\label{def:small tail value}
Let $\eta\in (0,1]$ and $\kappa\in [\underline{c},\bar{c}]$. 
A principal-agent instance with types supported on $[\underline{c},\,\overline{c}]$ has \emph{$(\kappa,\eta)$-thin-tail} if
\begin{eqnarray*}
\wel_{[\underline{c},\,\kappa]} \leq 
(1-\eta)\cdot\wel_{[\underline{c},\,\overline{c}]}.
\end{eqnarray*}
\end{definition}

Intuitively, the $(\kappa,\eta)$-thin-tail parameterization quantifies how much of the welfare is concentrated around the low cost types (i.e., high ability types).
The larger~$\kappa$ is, and the closer $\eta$ is to~1, the thinner the tail and the further the setting from point mass.%
\footnote{Similar assumptions are adopted in the literature of sample complexity for revenue-maximizing auctions \citep[e.g.,][]{devanur2016sample},
to exclude the situation that types from the tail event contribute a large fraction of the optimal objective value.
Although the assumptions share a similar format, the logic behind them is different. 
In the sample complexity literature, the thin-tail parameterization is adopted to ensure that the revenue from the small probability event is small, 
since it can only be captured with small probabilities given finite samples. 
In our model, we require the thin-tail parameterization to distinguish point-mass like settings from settings with sufficient uncertainty.
}
Note that although the thin-tail parameterization depends on the whole principal-agent instance, 
it is very easy to verify whether it can be satisfied given any instance since the computation of welfare contribution only requires the computation of welfare maximizing action $i^*(r,c)$ for each type $c$, rather than solving any optimization problem subject to incentive constraints.

\paragraph{Approximation guarantee.}
The thin-tail parameterization drives the following 
approximation guarantee for linear contracts:

\begin{theorem} 
\label{thm:universal}
Let $q,\alpha\in(0,1),\eta\in(0,1-q)$ and suppose the agent's cost for quantile~$q$ is~$c_q$. 
For any principal-agent instance with $(\frac{c_q}{\alpha},\eta)$-thin-tail,
a linear contract with parameter $\alpha$ provides expected revenue that is an $\frac{1}{(1-\alpha)\eta q}$-approximation of the optimal welfare.
\end{theorem}

Since the agent is not screened under linear contracts, 
\cref{thm:universal} also implies that 
the value of screening is small if the principal-agent setting has a thin tail.
Moreover, our approximation result is stated with respect to the benchmark of optimal welfare. Our result immediately translates into an approximation guarantee with respect to the optimal revenue since welfare is always an upper bound on revenue.

We note that while the approximation guarantee of \cref{thm:universal} is parametric (i.e., depends on the parameters of the setting in a fine-grained manner), we view this generality as a key strength of the result, as it captures the essential features driving good approximation. 
Moreover, compared to environments without sufficient uncertainty, where the worst-case approximation can be linear in the number of actions, our results lead to a significant improvement, especially when the agent has a rich set of choices, i.e., when the number of actions is large.\footnote{In the extreme case with a continuum of actions, the worst-case approximation is unbounded without the thin-tails parameterization.}

\begin{figure}[t]
\begin{center}
\begin{tikzpicture}[scale = 0.9]

\draw[scale=0.5,->] (-0.2,0) -- (14, 0);
\draw[scale=0.5,->] (0, -0.2) -- (0, 10);

\draw[dashed] (0,2.3) -- (2.2,2.3);
\draw[dashed] (2.2,0) -- (2.2,2.3);

\fill[pattern={north east lines},pattern color=gray]
    (0,0) rectangle +(2.2,2.3);

\begin{pgfonlayer}{bg}
  \path[fill=gray!10] (0,0) -- (0,2.3) -- (2.2,2.3) -- (2.2,0) -- cycle;
\end{pgfonlayer}               

\draw [thick] plot [smooth, tension=0.5] coordinates {(0, 3.5) (2, 2.5) (4, 0.3) (5.2, 0)};


\draw[scale=0.5] (12,0) -- (12,0.2);
\draw[scale=0.5] (12, -0.6) node {$1$};

\draw[scale=0.5] (0,8.5) -- (0.2,8.5);
\draw[scale=0.5] (-0.6,8.5) node {$1$};
\draw[scale=0.5] (0,-0.6) node {$0$};

\draw (7, -0.3) node {$q$};
\draw (-0.3, 5) node {$\eta$};

\end{tikzpicture}













\end{center}
\vspace*{-12pt}
\caption{For fixed $\alpha$, say $\alpha = 1/2$, we plot (black curve) the fraction $\eta$ of welfare from types above $c_q/\alpha = 2 c_q$ for $q \in [0,1]$. This curve is decreasing, and the intercept with the $y$-axis is typically smaller than $1$ because of the division of $c_q$ by $\alpha$. The best approximation guarantee that can be obtained with a given $\alpha$ is proportional to the area $\eta q$ of the largest box that fits under the curve (gray rectangle). 
The larger this area the further away the setting from being point mass-like, and the better the approximation.}
\label{fig:thin-tail}
\end{figure}
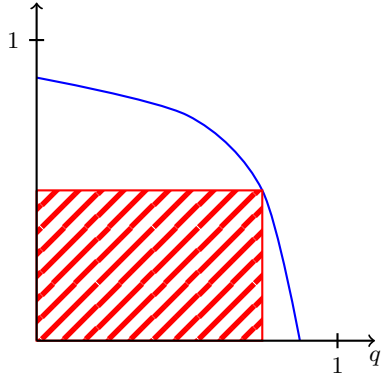

We provide a visual interpretation of the approximation guarantee from \cref{thm:universal} in \Cref{fig:thin-tail}.
Intuitively, the conditions in \cref{thm:universal} ensure that the welfare contribution from low types, i.e., types below $\frac{c_q}{\alpha}$, is sufficiently small compared to the optimal welfare. 
It is important to note that in order to guarantee constant approximations of linear contracts, it is sufficient to assume that the welfare is not concentrated around low cost (high ability) types. 
This is because high concentration of welfare around high cost (low ability) types cannot happen given any instance since for any set of high cost types with total measure $q<1$, 
the total welfare contribution from those types is at $q$ fraction of the total welfare. 


Our approximation guarantee in \cref{thm:universal} degrades smoothly as the tail in the principal-agent setting gets fatter. 
This implies that the relative performance of linear contracts improves when the uncertainty increases. 
We provide 
an alternative formalization
of this observation in \cref{appx:smoothed} via the smoothed analysis framework of~\citet{spielman2004smoothed}.


\paragraph{Proof outline.} The proof of \cref{thm:universal} utilizes a novel concept 
of being slowly-increasing (\cref{def:slow-inc-dist}), which can be applied to any cumulative distribution function and captures its rate of increase. 
In \cref{thm:slow} below we provide a parametric approximation factor for linear contracts
that depends on this rate of increase
as well as the parameters in the thin-tail condition.
\cref{thm:universal} follows immediately from \cref{thm:slow} by observing that any distribution is slowly-increasing for some parameters.

\subsection{Approximation for Slowly-Increasing Distributions}
\label{sub:slow-inc}

Any distribution $G$ is an increasing function, and 
\cref{def:slow-inc-dist} captures the rate of growth for function $G$.
Namely, if the cost increases by a factor of $\frac{1}{\alpha}$ from $\alpha c$ to $c$, 
then the CDF should increase by a multiplicative factor of at most $\frac{1}{\beta}$ (i.e., $G(c)\le \frac{1}{\beta} G(\alpha c)$). 
This condition becomes stronger the lower $\alpha$ is and the closer $\beta$ is to 1, and accordingly the approximation factor in the following theorem becomes better.
Note that we require the condition to hold only for types above~$\kappa$ for some parameter $\kappa\in[\underline{c},\bar{c}]$, 
that is, only for types that are not in the tail. 

\begin{definition}[Slowly-increasing distribution]\label{def:slow-inc-dist}
Let $\alpha,\beta\in (0,1]$ and $\kappa\in[\underline{c},\bar{c}]$.\footnote{For larger $\kappa$, it is easier for the slowly-increasing condition to be satisfied but harder for the thin-tail condition to be satisfied.} A distribution $G$ with support $[\underline{c},\,\overline{c}]$ is \emph{$(\alpha,\beta,\kappa)$-slowly-increasing} if
\begin{eqnarray*}
G(\alpha c) \geq \beta G(c) && \forall c\geq \kappa.
\end{eqnarray*}
\end{definition}

Intuitively, the slowly-increasing definition captures the spread of the type distribution. When the cumulative distribution function increases more slowly, the distribution is spread out over a wider range. 
Thus, this property serves as an additional measure on the uncertainty over the agent's ability. 
With this property, \cref{thm:slow} shows that linear contracts can generate revenue that is a constant approximation to the optimal welfare.


\begin{theorem}\label{thm:slow}
Suppose that \cref{asp:cost_structure} holds and fix $\alpha,\beta,\eta\in (0,1]$.
For any principal-agent instance, 
if there exists $\kappa\in[\frac{\underline{c}}{\alpha},\bar{c}]$ such that
the instance has $(\alpha,\beta,\kappa)$-slowly-increasing distribution of agent types
and has $(\kappa,\eta)$-thin-tail,
a linear contract with parameter $\alpha$ achieves expected revenue that is a $\frac{1}{(1-\alpha)\beta\eta}$-approximation to the optimal expected welfare.
\end{theorem}

The proof of \cref{thm:slow} is provided in \cref{apx:slow_increasing}. 
First, we prove \cref{thm:universal} by applying \cref{thm:slow}.
\begin{proof}[Proof of \cref{thm:universal}]
For any $q > 0$ and any $\alpha \in (0,1)$, 
let $\kappa$ be the constant such that $G(\alpha\kappa) = q$.
Note that by construction we have $\kappa \geq \frac{\underline{c}}{\alpha}$.
For any distribution $G$ and any cost $c\geq \kappa$, 
we have 
\begin{align*}
G(\alpha c) \geq G(\alpha\kappa) \geq G(\alpha\kappa)\cdot G(c)
\end{align*}
where the last inequality holds since $G(c)\leq 1$ for any $c$. 
Therefore, distribution $G$ satisfies 
$(\alpha, G(\alpha\kappa), \kappa)$-slow-increasing condition. 
Under the additional assumption of $(\kappa,\eta)$-thin-tail, 
by applying \cref{thm:slow}, 
the approximation ratio of linear contracts to optimal welfare is 
\begin{align*}
\frac{1}{\eta(1-\alpha) G(\alpha\kappa)} = \frac{1}{(1-\alpha) \eta q}
\end{align*}
and \cref{thm:universal} holds.
\end{proof}

Now we will provide the intuition behind the proof of \cref{thm:slow} and illustrate the role of both thin-tail parameterization and slowly-increasing condition in the analysis of constant approximations. 
From \citet{dutting2020simple}, we know that fixing any particular cost type $c$, linear contract cannot guarantee a constant approximation to the optimal revenue or optimal welfare for that type. The gap can in fact depend linearly on the number of possible actions in the setting. 
This is true even under our assumption of thin-tail and slowly-increasing distribution. 
However, in our model, there is uncertainty regarding the agent's ability, and the main idea of our proof is to show that for most types $c$, even though the revenue from type~$c$ can be small, there exists a lower cost type $c'<c$ such that the revenue contribution from type $c'$ under linear contracts approximately covers the welfare contribution from type~$c$. Thus in expectation, the revenue of the linear contract approximates optimal welfare. 

More specifically, we decompose the total welfare into two parts, i.e., 
\begin{align*}
\wel_{[\underline{c},\,\overline{c}]} = \wel_{[\underline{c},\,\kappa]} + \wel_{[\kappa,\,\overline{c}]}
\end{align*}
and show that the expected revenue from a linear contract with parameter $\alpha$ approximates each part separately.

We first provide the bound on $\wel_{[\kappa,\,\overline{c}]}$.
Our argument relies on the crucial observation that for any $\alpha\in [0,1]$ and any cost type $c\geq 0$, by offering a linear contract with parameter~$\alpha$, the chosen action of type $\alpha c$ is weakly higher than with the welfare optimal action for type $c$. 
This is formalized in the following lemma. 
\begin{lemma}\label{lem:covering_under_general_cost}
Under \cref{asp:cost_structure}, for any $c>0$ and any parameter $\alpha \in (0,1)$, 
\begin{align*}
i^*(\alpha r, \alpha c)\geq i^*(r,c).
\end{align*}
\end{lemma}
Intuitively, if the cost function were linear, i.e. $\zeta(c,\gamma)=c\gamma$, then the two actions would coincide, that is, $i^*(\alpha r, \alpha c)= i^*(r,c)$. 
However, as stated in \cref{asp:cost_structure}, it is more costly for higher cost types to choose higher effort actions. Therefore, the welfare-optimal action for type $c$ is weakly lower, i.e., $i^*(\alpha r, \alpha c)\geq i^*(r,c)$. 
The proof of \cref{lem:covering_under_general_cost} is provided in \cref{apx:slow_increasing}.

Note that in linear contract with parameter $\alpha$, the principal keep $1-\alpha$ fraction of the expected reward when type $\alpha c$ chooses action $i^*(\alpha r, \alpha c) \geq i^*(r, c)$. 
Therefore, the revenue from type $\alpha c$ is at least an $(1-\alpha)$-approximation to the optimal welfare from type~$c$ since the principal does not need to internalize the cost of the agent under linear contract. 
This argument applies for any $c\geq \kappa$ since the choice of $\kappa$ satisfies the requirement that $\kappa\geq \frac{\underline{c}}{\alpha}$, which implies that $\alpha c \geq \underline{c}$ is still in the support of the distribution. 
To compare the expected revenue and expected welfare, we apply the above argument for types above $\kappa$ and take the expectation. 
Note that in cases where the probability of $\alpha c$ is much smaller than $c$, we can also use the revenue contribution from cost types smaller than $\alpha c$ to cover the welfare contribution since the expected revenue from the agent given linear contracts is even higher if the cost type of the agent is lower. 
Our condition of slowly-increasing ensures that the total probability of lower cost types, i.e., total probability of types below $\alpha c$, is not far from the total probability of types below $c$. 
By carefully accounting for the expected revenue and welfare contributions, we show that the expected revenue from the linear contract provides a constant approximation to the welfare contribution of $\wel_{[\kappa,\,\overline{c}]}$.

For types around the lowest cost type $\underline{c}$, the above argument fails since there won't exist any lower types whose revenue can cover their welfare. However, the thin-tail parameterization implies that the welfare from those types is actually small compared to the total welfare, i.e., 
\begin{align*}
\wel_{[\underline{c},\,\kappa]} \leq 
(1-\eta)\cdot\wel_{[\underline{c},\,\overline{c}]}.
\end{align*}
Combining the arguments, we show that the expected revenue of linear contract with parameter $\alpha$ is a constant approximation to the optimal welfare of $\wel_{[\underline{c},\,\kappa]}$. 



\cref{thm:slow} allows us to obtain more refined approximation guarantees for broader classes of type distributions, 
e.g., distributions with non-increasing densities, 
which include uniform distributions and exponential distributions as special cases. 
The proof of \cref{cor:wel-implications} is derived by simple algebraic calculation, with details provided in \cref{appx:implications}.

\begin{corollary}
\label{cor:wel-implications}
For any $\underline{c}\geq 0$, a type distribution with \emph{non-increasing density} on $[\underline{c},\bar{c}]$ 
satisfies $(\frac{\kappa+1}{2\kappa},\frac{1}{2},\kappa\underline{c})$-slowly-increasing for any $\kappa > 1$.\footnote{Note that we can allow $\bar{c} = \infty$ and the domain of the density function is $[\underline{c},\infty$) in this case.}
Moreover, if the principal-agent setting has $(\kappa\underline{c},\eta)$-thin-tail, 
a linear contract achieves as revenue a $\frac{4\kappa}{\eta(\kappa-1)}$-approximation to the optimal welfare.
\end{corollary}

For example, when the principal-agent instance has $(3\underline{c},\frac{2}{3})$-thin-tail, 
i.e., the welfare contribution from types at most triple of the lowest type does not exceed one third of the optimal welfare, 
the revenue of the optimal linear contract is at least $\frac{1}{9}$ of the optimal welfare
if the type distribution has non-increasing density.
Moreover, in the special case of $\underline{c}=0$, 
$(0,1)$-thin-tail is satisfied for \emph{all} principal-agent instances. 
Therefore, a linear contract achieves a $4$-approximation to the optimal welfare
for distributions with non-increasing densities supported on $[0,\infty)$. 
The example with $\underline{c}=0$ is special because even if the types are concentrated around low values, the principal can incentivize the agent to exert high effort with almost zero payments, leading to revenue close to the optimal welfare.


\subsection{Additional Discussions}
\paragraph{Discussion of the thin-tail parameterization.} 
As noted earlier, our thin-tail parameterization is not only a property of the type distribution, but also of the principal-agent setting itself. We show here that a condition based solely on the distribution is insufficient for excluding point-mass-like behavior. 
This is because the instance can be such that it is too costly for any low ability type to choose any high effort action, which leads to almost negligible welfare contribution from those types. 
This makes the setting analogous in nature to a point-mass setting where the agent's type is $\underline{c}$. 
We illustrate this idea in a principal-agent setting when the type distribution is a uniform distribution. 

\begin{example}\label{exp:thin_tail}
Consider a principal-agent setting when the type distribution is uniform in $[1, 100]$. 
The cost function is linear, i.e., $\zeta(c,\gamma) = c\gamma$.
Suppose there is a continuum of actions $i\in[0,\infty)$ and each action $i$ requires effort $\gamma_i = \xi\cdot i^{\lambda}$ and generates expected reward $R_i=i$ where $\xi>0$ and $\lambda>1$.\footnote{Although our main model considers contract settings with finite number of actions, all the results extend for a continuum of actions, and we consider a continuum of actions here to simplify the computation.}
We do not explicitly specify the detailed distribution $F_i$ except for its expectation $R_i$ since it does not matter for either the optimal welfare or the revenue from linear contracts. 
\end{example}
In this example, the parameter $\lambda>1$ reflects the rate of increase for required units of effort for marginally increasing the expected output. 
It is easy to verify that the welfare optimal action for type $c$ is $i^*(r,c) = (\xi\lambda c)^{-\frac{1}{\lambda-1}}$ with expected welfare
\begin{align*}
\wel_c = \rbr{1-\frac{1}{\lambda}}\cdot(\xi \lambda c)^{-\frac{1}{\lambda-1}}.
\end{align*}
Note that both the welfare optimal action and the welfare contribution from type $c$ decreases dramatically as $c$ increases when $\lambda$ is close to $1$. 
That is, the constructed principal-agent instance has fat tail when $\lambda$ is close to $1$. 
For example, when $\lambda = 1.01$ and $\kappa = 2$, 
we have $\wel_{[1,2]} \geq 0.998\cdot\wel_{[1,100]}$, i.e., parameter $\eta \leq 0.2\%$ in this case.

Moreover, it is easy to compute that the optimal linear contract has parameter $\alpha = \frac{\lambda}{2\lambda-1}$. The multiplicative gap between the optimal welfare and the expected revenue from linear contract with parameter $\alpha = \frac{\lambda}{2\lambda-1}$ is 
$\frac{2\lambda-1}{\lambda-1}\cdot (\frac{\lambda}{2\lambda-1})^{-\frac{\lambda}{\lambda-1}}$, which can be unbound when $\lambda$ is sufficiently close to 1. 
For example, the expected revenue from linear contract is at most $0.4\%$ of the optimal welfare when $\lambda=1.01$. 
In \cref{apx:thin_tail}, we will provide a detailed construction such that the gap between the optimal revenue and the revenue of linear contracts is also large when the thin-tail parameterization is not satisfied. 

In contrast, when $\lambda$ is large, the principal-agent instance has thin tails. In this case, by simple calculation, we can show that the expected revenue from linear contract is at least $18.5\%$ of the optimal welfare if $\lambda\geq 3$
and at least $24\%$ of the optimal welfare if $\lambda\geq 20$.

\paragraph{Interpretation of approximation results.}

The approximation guarantees that come out of the thin-tail parameterization may appear large at first glance.
However, the benchmark here is the optimal welfare, which might be significantly higher than the optimal revenue, 
and hence the actual loss compared to the optimal revenue might be smaller. 
In fact, in \cref{sub:lin-bound} we provide an improved approximation ratio for distributions with non-increasing densities 
when compared to the optimal revenue
under a slightly different thin-tail parameterization.

Moreover, the interpretation of the approximation result is not to take the constant factor too literally \citep{hartline2012approximation,roughgarden2019approximately}.
Instead, we should focus on the relative comparison of the approximation factors in different environments.
In our model, the worst-case approximation ratio with degenerate pointmass type distributions degrades linearly with the number of actions available to the agent,
while the worst-case approximation ratio with sufficient uncertainty on agent's ability remains constant regardless of the number of actions.
This illustrates that sufficient uncertainty in agent's ability is the main factor driving the approximate optimality of linear contracts in our Bayesian model.

Another important comparison is between the approximation guarantees of different contract formats under the same assumption of sufficient uncertainty. As we illustrate later in \cref{sec:approx_other}, the 
approximation factors of other simple contracts, such as debt contracts, may degrade linearly with the number of actions available to the agent, even with sufficient uncertainty. This comparison highlights the robust performance of linear contract when there is sufficient uncertainty regarding the agent's ability. 
The main idea for this difference is that when the agent is offered with a linear contract, the revenue contribution from high ability agents approximately covers the welfare from low ability agents. 
Such covering argument fails for other contract formats such as debt contracts.

\section{Improved Approximation to Revenue Benchmark}
\label{sec:rev}

In this section, we refine the approximation guarantees of linear contracts by directly comparing them to the benchmark of optimal revenue (or, more precisely, a better proxy for it). 
For tractability, we assume throughout the remainder of the paper that the cost function is linear, specifically, $\zeta(c, \gamma) = c \cdot \gamma$.

Our argument builds on the standard notion of a virtual cost function $\virtual(c) = c + \frac{G(c)}{g(c)}$, and more generally the ironed version of it $\ironed(c)$. It further uses that the revenue achievable by any IC contract is upper bounded by the (ironed) virtual welfare. We provide a formal derivation of these facts along with additional observations in  Appendix~\ref{sub:primal}.

\subsection{Approximations for Linearly-Bounded Distributions}\label{sub:lin-bound}

Similar to our bounds for the welfare benchmark from the previous section, our approximation guarantees in this section are formulated in terms of a condition on the cost distribution and a suitable thin-tail parameterization.
Our distributional condition 
captures how well the virtual cost is approximated by a linear function. 
We show that ``the more linear'' the virtual cost, the better a linear contract approximates the optimal expected virtual welfare (and hence the optimal revenue). 
This is formalized in Definition~\ref{def:lin-bounded-dist}, where the ironed virtual cost of the distribution is sandwiched between two parameterized linear functions of the cost. 

\begin{definition}[Linearly-bounded distribution]\label{def:lin-bounded-dist}
Let $\beta \leq \alpha\in (0,1]$. A distribution $G$ with support $[\underline{c},\bar{c}]$ satisfies \emph{$(\alpha,\beta,\kappa)$-linear boundedness} if the corresponding ironed virtual value function satisfies
\begin{eqnarray*}
\frac{1}{\alpha}  c \leq  \ironed(c) \leq \frac{1}{\beta}  c   && \forall c \geq \kappa,\label{eq:lin-bounded-dist}
\end{eqnarray*}
where $\kappa \in [\underline{c},\bar{c}].$
\end{definition}

As pointed out by \citet{bulow1989simple}, the ironed virtual value function $\ironed(c)$ can be interpreted as the marginal cost function for incentivizing agent with type $c$ to choose the desirable action. 
Our assumption of linear boundedness holds if the marginal cost function has bounded rate of change. 

Our thin-tail parameterization is defined with respect to the virtual welfare rather than welfare.
Let $\vwel_{[a,b]}\triangleq\int_{a}^{b} R_{i^*(r,\ironed(c))}-\gamma_{i^*(r,\ironed(c))} \ironed(c) \dd G(c)$ be the expected virtual welfare from types within $[a,b]$. 

\begin{definition}[Thin-tail parameterization for virtual welfare]\label{def:small tail virtual}
A principal-agent instance with types supported on $[\underline{c},\,\overline{c}]$ has \emph{$(\kappa,\eta)$-thin-tail for virtual welfare} if
\begin{eqnarray*}
\vwel_{[\kappa,\,\overline{c}]} \geq \eta\cdot\vwel_{[\underline{c},\,\overline{c}]}.
\end{eqnarray*}
\end{definition}

The linearly-bounded condition together with the thin-tail parameterization enable the following parametric approximation guarantees.

\begin{theorem}\label{thm:lin-bounded-const-apx}
Fix $\alpha,\beta,\eta\in (0,1]$.
For any principal-agent instance and $\kappa$ such that
the instance has a $(\alpha,\beta,\kappa)$-linearly bounded distribution of agent types and a $(\kappa,\eta)$-thin-tail for virtual welfare, the linear contract with parameter $\alpha$ achieves expected revenue that is:
\begin{enumerate}
    \item a $\frac{1}{\eta (1-\alpha)}$-approximation to the optimal expected virtual welfare.\label{item:without-highest-cost}
    \item \label{item:with-highest-cost} a $\frac{1-\beta}{\eta (1-\alpha)}$-approximation to the optimal expected virtual welfare if $i^*(r,\ironed(\bar{c}))=0$.%
\end{enumerate}
\end{theorem}

The additional condition $i^*(r,\ironed(\bar{c}))=0$ in \cref{thm:lin-bounded-const-apx}.\ref{item:with-highest-cost}
assumes that the highest cost type is sufficiently large such that the principal cannot incentivize this type to exert costly effort. 
This further ensures that the range of support for the cost distribution is sufficiently large and hence the approximation guarantee is improved. 

We provide the detailed proof of \cref{thm:lin-bounded-const-apx}.\ref{item:without-highest-cost} to illustrate the main ideas, which is similar the proof of \cref{thm:slow} that uses the revenue from high ability agents to cover the welfare of low ability agents. The main difference is that now we apply the argument in the virtual value space and we wish to cover the virtual welfare of low ability agents.
The proof of \cref{thm:lin-bounded-const-apx}.\ref{item:with-highest-cost} with refined approximations is provided in \cref{apx:linear bound proof}.

\begin{proof}[Proof of \cref{thm:lin-bounded-const-apx}.\ref{item:without-highest-cost}]
Denote by $\vwel$ the optimal expected virtual welfare, and by $\apx$ the expected revenue of a linear contract with parameter $\alpha$; our goal is to show that $\apx \geq \eta(1-\alpha) \vwel$, which suffices for our purpose since \cref{cor:upper-bound} implies that $\vwel$ is an upper bound on the optimal revenue.
Recall $r$ represents the vector of rewards given different actions. 
Let $i^*(\alpha r, c)$ be the action that agent with type $c$ chooses given a linear contract with parameter~$\alpha$,
and let $i^*(r,\ironed(c))$ be the virtual welfare maximizing action for agent with type $c$. 
We first show that under $(\alpha,\beta,\kappa)$-linear bounded assumption, 
\begin{align*}
i^*(\alpha r, c)\geq i^*(r,\ironed(c)),\quad\forall c \geq \kappa.
\end{align*} 
To do so, denote by $\hat{i}=i^*(r,\ironed(c))$, we will first show that $\alpha R_{\hat{i}}-\gamma_{\hat{i}}c \geq \alpha R_{\ell}-\gamma_{\ell}c$ for any $\ell\leq\hat{i}$. 
Note that by the definition of $\hat{i}$,
since $\ell$ is a feasible action choice, 
we have that $R_{\hat{i}}-\gamma_{\hat{i}}\ironed(c) \geq R_{\ell}-\gamma_{\ell}\ironed(c)$. 
That is, $R_{\hat{i}}-R_{\ell} \geq (\gamma_{\hat{i}}-\gamma_{\ell})\ironed(c)$ for any $\ell\leq\hat{i}$. 
Combining the above with $\ironed(c)\geq \frac{1}{\alpha}c$ for any $c\geq \kappa ,$ we have $\alpha R_{\hat{i}}-\alpha R_{\ell} \geq (\gamma_{\hat{i}}-\gamma_{\ell})c$ for any $\ell\leq\hat{i}$ and any $c\geq \kappa$. 
Therefore, $\alpha R_{\hat{i}}-\gamma_{\hat{i}}c \geq \alpha R_{\ell} -\gamma_{\ell}c$ for any $\ell\leq\hat{i}$ and $c\geq \kappa$.
This implies that $i^*(\alpha r, c)\geq \hat{i} = i^*(r,\ironed(c))$ for any $c \geq \kappa$. 
Since the expected reward is monotone in the chosen action, we have
\begin{align*}
R_{i^*(\alpha r, c)}\geq R_{i^*(r,\ironed(c))}, \quad\forall c \geq \kappa.
\end{align*}

We now use this to lower bound $\apx$. 
Recall that $\apx=\int_{\underline{c}}^{\bar{c}}g(c)(1-\alpha)R_{i^*(\alpha r,c)}\dd c$. Since $R_i \geq 0$ $\forall i \in [n]$, $\apx\geq \int_{\kappa }^{\bar{c}}g(c)(1-\alpha)R_{i^*(\alpha r,c)}\dd c$. 
By the above arguments we have that 
\begin{align}\label{eq:apx-only-alpha}
\apx\geq (1-\alpha)\int_{\kappa }^{\bar{c}}g(c)R_{i^*(r,\ironed(c))}\dd c.
\end{align}
Recall that $\vwel_{[\kappa ,\bar{c}]}=\int_{\kappa }^{\bar{c}}g(c)(R_{i^*(r,\ironed(c))}-\gamma_{i^*(r,\ironed(c))}\ironed(c))\dd c$. Since $\gamma_{i}\geq 0$ for any $i\in [n]$ and $\ironed(c)\geq 0$  for any $c\in [\underline{c},\bar{c}]$,
the latter is bounded from above as $\int_{\kappa }^{\bar{c}}g(c)R_{i^*(r,\ironed(c))}\dd c$. 
Combining this with \Cref{eq:apx-only-alpha}, we have that 
\begin{equation*}
\apx\geq (1-\alpha) \vwel_{[\kappa ,\bar{c}]} \geq \eta (1-\alpha)\vwel.\qedhere
\end{equation*}
\end{proof}

By \cref{thm:lin-bounded-const-apx} we obtain improved approximation ratios for several families of commonly observed distributions 
when compared to the optimal revenue, 
with detailed algebraic calculation provided in \cref{appx:implications}.

\begin{corollary}
\label{cor:rev-implications}
For any principal-agent setting, the following holds:
\begin{enumerate}
    \item If the type distribution has \emph{non-increasing density} on $[\underline{c},\bar{c}]$ for $\bar{c}>\underline{c}\geq 0$, then
    $(\frac{\kappa}{2\kappa-1},0,\kappa\underline{c})$-linear boundedness is satisfied for any $\kappa > 1$.
    With the parameterization of $(\kappa\underline{c},\eta)$-thin-tail for virtual welfare,
    a linear contract achieves a $\frac{2\kappa -1}{\eta (\kappa-1)}$-approximation to the optimal revenue.\label{item:rev-non-increasing-dense}
    \item If the type distribution is \emph{uniform} $U[0,\bar{c}]$, then $(\frac{1}{2},\frac{1}{2},0)$-linear boundedness 
    and $(0,1)$-thin-tail for virtual welfare are satisfied. A linear contract is revenue optimal when $i^*(r,\ironed(\bar{c}))=0$.\footnote{Note that this generalizes and strengthens one of the main findings in \citet{AlonDT21}.}\label{item:rev-uniform}
    \item If the type distribution is \emph{normal} $\mathcal{N}(\mu,\sigma^2)$ truncated at $c=0$ with $\sigma\geq \frac{5}{2\sqrt{2}}\mu$, then $(\frac{2}{3},0,0)$-linear boundedness 
    and $(0,1)$-thin-tail for virtual welfare are satisfied. 
    A linear contract achieves a $3$-approximation to the optimal revenue.\label{item:rev-normal}
\end{enumerate}
\end{corollary}
Again, in the special case of $\underline{c}=0$, 
$(0,1)$-thin-tail for virtual welfare is satisfied for all principal-agent instance. 
Therefore, a linear contract achieves a $2$-approximation to the optimal revenue for distributions with non-increasing densities supported on $[0,\bar{c}]$. 
This improves the worst case approximation of $4$ when compared to the optimal welfare. 
Moreover, our method allows us to establish the exact optimality of linear contracts for uniform distribution 
when the types are sufficiently dispersed, i.e., when $i^*(r,\ironed(\bar{c}))=0$. 
We also obtain constant approximation bounds for truncated normal distributions 
when the variance is large, a commonly adopted metric for measuring the amount of uncertainty in type distributions.\footnote{
\citet{gottlieb2022simple} give a condition (e.g., binary action) under which debt contracts are optimal. This is not in contrast to our result establishing the optimality of linear contracts for uniform distributions, because their result requires there to be a zero cost action that leads to positive reward (see Footnote 8 of their paper), which is not the case in our model.}

\subsection{Tight Approximation for General Distributions}
\label{sub:apprx-n}



We present a tight $\Theta(n)$-approximation guarantee for linear contracts for general distributions without thin-tail parameterizations. 
The lower bound of $\Omega(n)$ is illustrated in \cref{ex:scaling-setting} for uniform type distributions,
and our upper bound provides a tight approximation guarantee for any type distribution.
Our bounds imply that a linear contract achieves a constant approximation to the optimal revenue when the number of actions is a constant. 
This is in sharp contrast to the multi-dimensional type model of \citet{guruganesh20} and \citet{CastiglioniMG22}, where the approximation ratio is unbounded even for a constant number of actions.\footnote{In the multi-dimensional type model, there exists heterogeneity in the cost of different actions given different types. For example, there may exist one type that has lower cost for action 1 and another type that has lower cost for action 2.}

\begin{theorem}\label{thm:upper-bound-n}
For any principal-agent instance with $n$ actions, a linear contract achieves expected revenue that is an $n$-approximation to the optimal revenue.
\end{theorem}

The proofs appear in Appendix~\ref{appx:apprx-n}. The proof of Theorem~\ref{thm:upper-bound-n} relies on similar insights on the sets of breakpoints as in the proof of Theorem~\ref{thm:lin-bounded-const-apx}.




\section{Suboptimality of Other Simple Contracts}
\label{sec:approx_other}
In addition to linear contracts, there exist other simple forms of contract such as debt contract \citep[e.g.,][]{innes1990limited,gottlieb2022simple} or single-outcome payment contracts \citep[e.g.,][]{FII-23,dutting2024ambiguous}. 
In this section, we show that the worst case approximations of those simple contracts is linear in the number of actions even  if the thin-tail condition is satisfied. 
Certainly there may exist other simple contract formats  worth consideration, but we do not wish to provide an exhaustive list of results showing the approximations of all simple contracts in this section. 
The main purpose is to provide an illustration of the failure of approximations for contracts that do not take a linear format even when there is sufficient uncertainty regarding 
the agent's ability. 
This further motivates the use of linear contracts in practical applications. 

\subsection{Debt Contracts}
\label{apx:debt}

In this section, we consider debt contracts where the principal specifies a face value $\bar{r}$ such that the principal retains the minimum between the realized reward and the face value for each outcome. 


\begin{definition}
A contract $(x,t)$ is a \emph{debt contract} with face value $\bar{r}$ if for any outcome $j\in[m]$ and any cost $c\in C$, the payment function satisfies $t^{c}_j = \max\{r_j-\bar{r},0\}$.\footnote{In our definition, the principal offers the same face value for any reported cost $c$. This is because if the principal offers different face values for different costs, the agent always maximizes his expected utility by reporting the cost that minimizes the face value.}
\end{definition}

\begin{proposition}\label{prop:debt}
For any $n\geq 2$, let $q,\alpha\in(0,1),\eta\in(0,1-q)$, and suppose the agent's cost for quantile~$q$ is $c_q$. 
There exists a principal-agent instance with $n$ actions that satisfies $(\frac{c_q}{\alpha},\eta)$-thin-tail 
and the approximation ratio of debt contracts is at least $\Omega(n)$.
\end{proposition}

The thin-tail condition in \cref{prop:debt} is the same as in \cref{thm:universal}. 
We show that even under the same thin-tail condition, the constant approximation for debt contracts is not guaranteed. 
The proof of \cref{prop:debt} relies on the following example. 
\begin{example}
For any $n\geq 2$ and $\epsilon\in (0,\frac{1}{2})$, consider an instance with $m=n+1$. 
Let $0$ be the null action and $r_0=0$ be the null outcome assumed in \cref{sec:prelim}.
For any outcome $j\in[m]\backslash\{0\}$, the reward for outcome $j$ is $r_j = \epsilon^{1-j}$. 
For any action $i\in [n]\backslash\{0\}$, the distribution over outcomes satisfies 
\begin{align*}
F_{i,j} = \begin{cases}
i\cdot\epsilon^{m-i} & j=m-i+1,\\
1-i\cdot\epsilon^{m-i} & j=1,\\
0 & \text{otherwise.}
\end{cases}
\end{align*}
Moreover, for any action $i\in [n]\backslash\{0\}$, 
the per-unit level of effort is $0<\gamma_1<\cdots<\gamma_n<\epsilon$.
\end{example}

In this example, it is easy to verify that the expected rewards for action $i\in[n]\backslash\{0\}$ is $R_i = 1+i(1-\epsilon^{m-i})$.

Next, for any action $i\in [n]\backslash\{0,1\}$, we provide a necessary condition on the face value~$\bar{r}$ for the principal to incentivize the agent to choose action $i$. 
It is immediate that $\bar{r}$ must be at most $\epsilon^{-(m-i)}$.
Moreover, one necessary condition is that the expected payment $T_i$ given action $i-1$ is at most the expected payment $T_{i-1}$ given action $i-1$ since the required per-unit effort for action $i-1$ is lower. This implies that for any $\bar{r} > 1$,
\begin{align*}
&i\cdot\epsilon^{m-i}\cdot (\epsilon^{-(m-i)}-\bar{r}) 
> (i-1)\cdot\epsilon^{m-i+1}\cdot (\epsilon^{-(m-i+1)}-\bar{r})\\
\Leftrightarrow\,& \bar{r} < \frac{1}{(i-(i-1)\epsilon)\cdot \epsilon^{m-i}}.
\end{align*}
Therefore, the expected revenue of the principal for incentivizing the agent to choose action $i\geq 2$ under debt contract is at most 
\begin{align*}
1-i\cdot\epsilon^{m-i} + \bar{r} \cdot i\cdot\epsilon^{m-i}
\leq 1+\frac{i}{(i-(i-1)\epsilon}\leq 3.
\end{align*}
Since the social welfare for incentivizing the agent to choose action $i\in\{0,1\}$ is at most 2, 
the expected revenue of the principal is at most 3 given any debt contract. 
This holds for arbitrary distribution over agent's types. 

However, for any distribution over agent's types with bounded support, there exists a sufficient small $\epsilon>0$ such that the cost of the highest type for choosing action $n$ is at most~$1$. 
It is easy to verify that the thin-tail parameterization can be satisfied with sufficient small $\epsilon>0$ and uniform type distribution with sufficiently large support.
Moreover, by using a contract that provides a positive reward $\frac{1}{n\epsilon}$ only on outcome $2$, 
the agent is incentivized the choose action $n$ for all types, which guarantees expected revenue at least $n-1$ for the principal. 
Therefore, the multiplicative gap between optimal contract and debt contract is at least $\frac{n-1}{3} = \Omega(n)$.

\subsection{Single-Outcome payment Contracts}
\label{apx:single_reward}
In this section, we consider single-outcome payment contracts where the principal only specifies a positive payment to the agent on a single outcome. 
In classic models with binary actions and without private types, it is folklore that single-outcome payment contracts are optimal by providing a positive reward to the agent on the outcome that maximizes the likelihood ratio.
Moreover, \citet{dutting2024ambiguous} show that a collection of single-outcome payment contracts is optimal when the designer can commit to ambiguous contracts. 
\begin{definition}
A contract $(x,t)$ is a \emph{single-outcome payment contract} if there exists $j\in[m]$ and reward $\bar{r}>0$ such that for any cost $c\in C$, the payment function satisfies $t^{c}_{j'} = \bar{r}$ if $j'=j$ and $t^{c}_{j'} = 0$ otherwise.
\end{definition}
\begin{proposition}\label{prop:single}
For any $n\geq 2$, let $q,\alpha\in(0,1),\eta\in(0,1-q)$ and suppose the agent's cost for quantile~$q$ is $c_q$. 
There exists a principal-agent instance with $n$ actions that satisfies $(\frac{c_q}{\alpha},\eta)$-thin-tail 
and the approximation ratio of single-outcome payment contracts is at least $\Omega(n)$.
\end{proposition}
The proof of \cref{prop:single} relies on the following example. 
\begin{example}
For any $n\geq 2$ and any $\epsilon>0$, consider an instance with $m=n$. 
Let $0$ be the null action and $r_0=0$ be the null outcome assumed in \cref{sec:prelim}.
For any outcome $j\in[m]\backslash\{0,1\}$, the reward for outcome $j$ is $r_j = 1$ and we let $r_1=0$.\footnote{We can also perturb the rewards for outcomes $j\geq 2$ such that all those rewards are distinct but sufficiently close to 1.} 
For action $n$, the outcome is drawn from a uniform distribution in $\{1,\dots,m-1\}$. 
For any action $i\in [n]\backslash\{0,n\}$, the distribution over outcomes satisfies 
\begin{align*}
F_{i,j} = \begin{cases}
\frac{2}{n} & j=i,\\
1-\frac{2}{n} & j=1,\\
0 & \text{otherwise.}
\end{cases}
\end{align*}
Moreover, for any action $i\in [n]\backslash\{0\}$, 
the per-unit level of effort is $0<\gamma_1<\cdots<\gamma_n<\epsilon$.
\end{example}

In this example, it is easy to verify that the expected rewards for action $i\in[n]\backslash\{0,n\}$ is $R_i = \frac{2}{n}$, 
and the expected rewards for action $n$ is $R_n = 1$.
Note that for any $j\in [m]$ and any single-outcome payment contract that provides a positive reward only for outcome $j$, 
there exists an action $i<n$ ($i=j-1$ if $j\neq 1$ and $i=j$ if $j=1$) such that the expected payment $T_i$ given action $i$ is at least the expected payment $T_n$ given action $n$. 
Therefore, the agent cannot be incentivized to choose action $n$ given any single-outcome payment contract. 
In this case, the expected revenue of the principal is at most $\max_{i\leq n-1} R_i \leq \frac{2}{n}$. 
This bound holds for any distribution over types. 

However, given any distribution over types with bounded support, there exists a sufficiently small $\epsilon$ such that the agent is incentivized to choose action $n$ given linear contract with parameter $\alpha=\frac{1}{2}$ for all types. 
The expected revenue of the principal in this case is~$\frac{1}{2}$. 
Therefore, the multiplicative gap between the optimal contract and the single-outcome payment contract is at least $\frac{n}{4} = \Omega(n)$.

\appendix

\section{Near-Optimality of Linear Contracts}\label{appx:linear}

In this section we provide the proofs (or the parts of the proofs) omitted from \cref{sec:linear} on the near-optimality of linear contracts.

\subsection{Slowly-Increasing Distributions}
\label{apx:slow_increasing}
\begin{proof}[Proof of \cref{lem:covering_under_general_cost}]
To prove the lemma, it is sufficient to show that for any $c>0$ and any parameter $\alpha \in (0,1)$, 
the expected utility of type $\alpha c$ for choosing any action $i < i^*(r,c)$ given linear contract with parameter $\alpha$ is at most his utility for choosing action $i^*(r,c)$. 
This implies that there exists an action weakly larger than $i^*(r,c)$ that is a weakly best response for the agent. 
First note that 
\begin{align}\label{eq:larger_cost_diff}
&\alpha\cdot(\zeta(c,\gamma_{i^*(r,c)}) - \zeta(c,\gamma_i))
\geq \alpha \cdot\zeta(c,\gamma_i) \cdot \rbr{\frac{\zeta(\alpha c,\gamma_{i^*(r,c)})}{\zeta(\alpha c,\gamma_i)} - 1} \\
& = \alpha \cdot \frac{\zeta(c,\gamma_i)}{\zeta(\alpha c,\gamma_i)} \cdot \rbr{\zeta(\alpha c,\gamma_{i^*(r,c)}) - \zeta(\alpha c,\gamma_i)}
\geq \zeta(\alpha c,\gamma_{i^*(r,c)}) - \zeta(\alpha c,\gamma_i),\nonumber
\end{align}
where the first inequality holds by applying \cref{eq:equiv_ineq}
and the last inequality holds by applying condition 1 and 2 of \cref{asp:cost_structure}, i.e., $\zeta(\alpha c,\gamma_{i^*(r,c)}) - \zeta(\alpha c,\gamma_i)$ is non-negative and $\frac{\zeta(c,\gamma_i)}{\zeta(\alpha c,\gamma_i)} \geq \frac{1}{\alpha}$.

Moreover, by the definition of $i^*(r,c)$, we have 
$R_i - \zeta(c,\gamma_i) \leq R_{i^*(r,c)} - \zeta(c,\gamma_{i^*(r,c)})$, which implies that 
\begin{align*}
\alpha\cdot(R_{i^*(r,c)} - R_i) \geq \alpha\cdot(\zeta(c,\gamma_{i^*(r,c)}) - \zeta(c,\gamma_i)) \geq \zeta(\alpha c,\gamma_{i^*(r,c)}) - \zeta(\alpha c,\gamma_i)
\end{align*}
where the last inequality holds by applying \cref{eq:larger_cost_diff}. 
By reordering the terms, the above inequality implies that the expected utility of type $\alpha c$ for choosing any action $i < i^*(r,c)$ given linear contract with parameter $\alpha$ is at most his utility for choosing action $i^*(r,c)$.  
\end{proof}

\begin{proof}[Proof of \cref{thm:slow}] 
Denote by $\wel$ the optimal expected welfare, and by $\apx$ the expected revenue of a linear contract with parameter $\alpha$; we show that $\apx \geq (1-\alpha)\beta\eta\cdot \wel$. 

First, letting $\zeta_1(c,\gamma_{i^*(r,c)})$ be the partial derivative of $\zeta$ on its first coordinate, the welfare from types above $\kappa$ is\footnote{Although the support of the distribution $G$ is $[\kappa,\bar{c}]$, we still take the integration from $c$ to $\infty$. 
Note that to make things well defined, we have $g(c)=0$
and $G(c)=1$ for any $c> \bar{c}$.}
\begin{align*}
\wel_{[\kappa,\,\bar{c}]}
&=\int_{\kappa}^{\infty} R_{i^*(r,c)} - \zeta(c,\gamma_{i^*(r,c)}) \dd G(c)\\
&= -G(\kappa) \cdot (R_{i^*(r,\kappa)} - \zeta(\kappa,\gamma_{i^*(r,\kappa)})
+ \int_{\kappa}^{\infty} G(c) \cdot \zeta_1(c,\gamma_{i^*(r,c)}) \dd c,
\end{align*}
where the equality holds by integration by parts 
and the fact that the derivative of $R_{i^*(r,c)} - \zeta(c,\gamma_{i^*(r,c)})$
with respect to $c$ is $\zeta_1(c,\gamma_{i^*(r,c)})$ (by the Envelope Theorem). 

Moreover, the revenue of $\apx$ is 
\begin{align*}
\apx&= \int_{\underline{c}}^{\infty} (1-\alpha) R_{i^*(\alpha r,c)} g(c)\dd c
\geq \int_{\underline{c}}^{\infty} (1-\alpha) R_{i^*(r,\frac{c}{\alpha})} g(c)\dd c\\
&\geq (1-\alpha)  \int_{\underline{c}}^{\infty} g(c) ( R_{i^*(r,\frac{c}{\alpha})} - \zeta(\frac{c}{\alpha},\gamma_{i^*(r,\frac{c}{\alpha})})) \dd c\\
&\geq (1-\alpha)\alpha  \int_{\frac{\underline{c}}{\alpha}}^{\infty} g(\alpha c) ( R_{i^*(r,c)} - \zeta(c,\gamma_{i^*(r,c)})) \dd c,
\end{align*}
where the first inequality holds by applying \cref{lem:covering_under_general_cost},
and the second inequality holds simply because the cost is non-negative.
The last equality holds by a simple change of variable. 
Using integration by parts: 
\begin{align*}
\apx &\geq 
(1-\alpha) \int_{\frac{\underline{c}}{\alpha}}^{\infty} G(\alpha c) \cdot \zeta_1(c,\gamma_{i^*(r,c)}) \dd c
\geq (1-\alpha)  \int_{\kappa}^{\infty} G(\alpha c) \cdot \zeta_1(c,\gamma_{i^*(r,c)}) \dd c.
\end{align*}
Since $G(\alpha c) \geq \beta G(c)$ for any $c\geq \kappa$, 
we have that 
\begin{align*}
\apx\geq (1-\alpha)\beta \cdot \wel_{[\kappa,\,\bar{c}]}
\geq (1-\alpha)\beta\eta\cdot \wel,
\end{align*}
where the last inequality holds by 
$(\kappa,\eta)$-thin-tail parameterization.
\end{proof}

\subsection{Discussion of Thin-tail Parameterization}
\label{apx:thin_tail}
\begin{example}\label{ex:scaling-setting}
Let $\delta < 1$ to be determined. The set of actions is $\{0,\dots,n\}$. 
The type $c$ is drawn from a uniform distribution in $[1, \bar{c}]$
and the cost function is linear, i.e., $\zeta(c,\gamma)=c\gamma$.
The required effort levels are $\gamma_0=0,\gamma_{i}=\frac{1}{\delta^{i}}-(i+1)+\delta i$ $\forall i\in [n]$. The set of outcomes is $\{0,\dots,n\}$. The rewards are given by $r_0=0$ and $r_i=\frac{1}{\delta^{i}}$ $\forall i\in [n]$. Outcome probabilities are $F_{i,i}=1$ $\forall i \in \{0,\dots,n\}$.\footnote{A similar example is presented in Theorem $3$ of \citet{dutting2020simple} to show that the gap between linear contracts and optimal contracts are large when the type is known.}
\end{example}

Suppose the type is $c=1$. 
In this case, it can be verified that for every action~$i$, the minimal parameter $\alpha_i$ of a linear contract that incentivizes the agent to take action $i$ is 
\begin{eqnarray*}
\alpha_i =\begin{cases}
1-2\delta+\delta^2 & i=1,\\
1-\delta^{i} & i>1.
\end{cases}
\end{eqnarray*}
Therefore, the expected revenue of action $i$ given a linear contract is at most $\frac{1}{\delta^i}\delta^i=1$ for $i>1$ and $(2\delta -\delta^2)\frac{1}{\delta}<2$ for $i=1$. That is, at most $2$. However, when not restricting herself to linear contracts, the principal can extract the entire welfare by paying $t_n=\gamma_n$ and incentivizing the agent to take action $n$. Then, the optimal revenue is at least $\frac{1}{\delta^n}-(\frac{1}{\delta^n}-(n+1)+\delta n)=n+1-\delta n.$ We can conclude that in this example, the multiplicative loss in the principal’s expected revenue from using a linear contract rather than an optimal one is at least $\Omega(n)$.

Now we go back to the example with uniform distributions. 
For any $\epsilon>0$, there exists $\delta>0$ such that given the above construction of the principal-agent setting, 
any type with cost above $1+\epsilon$ will have zero contribution to the welfare. 
In particular, 
\begin{eqnarray}\label{eq:saling-1}
  r_n-\gamma_n (1+\epsilon)= 0.
\end{eqnarray}
Let $p_{\epsilon}$ be the probability that a type is within $[1,1+\epsilon]$.
Using similar arguments to those presented above, a linear contract's revenue for any type $c$ is bounded by the minimum between~$2$ and the welfare of type $c$. 
Thus, the revenue of a linear contract is at most~$2p_{\epsilon}$.
However, by offering payment $t_n = (1+\frac{\epsilon}{2})\gamma_n$ and $t_j=0$ $\forall j \leq n-1$, all
agents with types $c\leq (1+\frac{\epsilon}{2})$ will choose action $n$. Each extracts a revenue of at least $\frac{1}{2}(n+1-\delta n)$. To see this, note that the expected revenue of each type is $r_n-(1+\frac{\epsilon}{2})\gamma_n.$ 
This equals $\frac{1}{2}(r_n-\gamma_n)+\frac{1}{2}(r_n-(1+\epsilon)\gamma_n)$. Since $\epsilon$ is chosen such that $r_n-\gamma_n(1+\epsilon)=0$ and $r_n-\gamma_n= n+1-\delta n$, the expected revenue is $\frac{1}{2}(n+1-\delta n)$. 
Therefore, the optimal revenue is at least $\frac{1}{2}p_{\epsilon}(n+1-\delta n)$, and hence,
regardless of $\bar{c}$,
the gap is again $\Omega(n)$ when $\delta$ is sufficiently small.

\section{Characterization of IC Contracts} 
\label{sub:primal}

We characterize the set of allocation rules that can be implemented by IC contracts when the cost function is linear (i.e., $\zeta(c, \gamma) = c\cdot \gamma$), 
and derive the unique expected payment of each type given any implementable allocation rule. 
\begin{definition}
\label{def:implementable}
An allocation rule $x$ is \emph{implementable} if there exists a payment rule $t$ such that contract $(x,t)$ is IC. 
\end{definition}
\begin{lemma}[Integral Monotonicity and Payment Identity]\label{lem:primal characterization}
An allocation rule $\alloc$ is implementable if and only if
for every type $c\in C$, there exist payments $t^c\in \mathbb{R}^{m+1}$ such that $i^*(t^c,c)=x(c)$ and\,\footnote{$\int_c^{c'}$ is defined as $-\int_{c'}^c$ when $c' < c$.} 
\begin{align}
\int_c^{c'} \gamma_{\alloc(z)} - \gamma_{i^*(t^c,z)} \dd z \leq 0, \qquad \forall c'\in C. \label{eq:curvature}
\end{align}
Moreover, the expected payment $T^c_{x(c)}$ given an implementable allocation rule $\alloc$
satisfies
\begin{align}
T^c_{x(c)} = \gamma_{\alloc(c)}\cdot c 
+ \int_{c}^{\bar{c}} \gamma_{\alloc(z)} \dd z 
+ (T^{\bar{c}}_{x(\bar{c})} - \bar{c}\cdot \gamma_{\alloc(\bar{c})}).\label{eq:pay-identity}
\end{align}
\end{lemma}
This characterization of implementable allocation rules can be thought of as the ``primal'' version of a recent ``dual'' characterization due to \citet{AlonDT21}.\footnote{\citet{castro2024disentangling} provides sufficient conditions for decoupling the incentive constraints into moral hazard and adverse selection constraints separately. Our characterization is orthogonal to theirs. 
Instead, our characterization is more similar to characterizations for single-dimensional type spaces in a different context of selling information 
\citep[e.g.,][]{li2021selling,yang2022selling}.}
Our characterization uses the standard Envelope Theorem \citep{MilgromS02}.
Note that our characterization of integral monotonicity implies that the convexity in agent's utilities is only necessary, but not sufficient for ensuring the incentive compatibility of the contracts.

\begin{proof}[Proof of \cref{lem:primal characterization}]
Let $u(c)$ be agent's utility for following the mechanism
and let $\util_c(z)$ be the utility function of the agent with type $z$ when he has reported $c$ to the principal. 
By the Envelope Theorem \citep{MilgromS02},
we have that for any $c,c'$, 
\begin{align*}
u(c') &= u(c) - \int_c^{c'} \gamma_{\alloc(z)} \dd z
\intertext{and}
u_c(c') &= u_c(c) - \int_c^{c'} \gamma_{i^*(t^c,z)} \dd z.
\end{align*}
Note that the agent with cost $c'$ has no incentive to deviate the report to $c$ 
if and only $u(c') \geq u_c(c')$, 
which is equivalent to 
\begin{align*}
\int_c^{c'} \gamma_{\alloc(z)} - \gamma_{i^*(t^c,z)} \dd z \leq 0, \qquad \forall c'
\end{align*}
since $u_c(c) = u(c)$. Now given any implementable allocation rule $\alloc$, 
the corresponding interim utility 
satisfies $u(c) = u(\bar{c}) + \int_c^{\bar{c}} \gamma_{\alloc(z)} \dd z$
and hence
\begin{align*}
T^c_{x(c)} &= \gamma_{\alloc(c)}\cdot c 
+ u(\bar{c}) + \int_c^{\bar{c}} \gamma_{\alloc(z)} \dd z\\
&= \gamma_{\alloc(c)}\cdot c 
+ \int_{c}^{\bar{c}} \gamma_{\alloc(z)} \dd z 
+ (T^{\bar{c}}_{x(\bar{c})} - \bar{c}\cdot \gamma_{\alloc(\bar{c})}). \qedhere
\end{align*}
\end{proof}

Let the virtual cost of agent with type $c$ be $\virtual(c) = c+\frac{G(c)}{g(c)}$.
We characterize the expected revenue of any IC contract in the form of virtual welfare by applying \cref{lem:primal characterization} and integration by parts.

\begin{corollary}
\label{cor:expected-rev}
For any implementable allocation rule x that leaves an expected utility of $u(\bar{c})$ to the agent with highest cost type $\bar{c}$, 
the expected revenue of the principal is
\begin{align*}
\expect[c\sim G]{R_{\alloc(c)} - \gamma_{\alloc(c)} \cdot \virtual(c)} - u(\bar{c}).
\end{align*}
\end{corollary}

One difference between \cref{cor:expected-rev} and the seminal Myerson's lemma in auction design is the former's dependence on the utility $u(\bar{c})$ of the highest cost type. 
Moreover, utility $u(\bar{c})$ cannot be set to $0$ without loss of generality in our model, as it can influence the feasibility of allocation rule~$\alloc$ through the complex non-local IC constraints
due to the limited liability requirements. 
This makes the expression for expected revenue in \cref{cor:expected-rev} harder to work with.
Moreover, the classic approach in mechanism design identifies the optimal allocation rule for maximizing the virtual welfare without incentive constraints
and then shows that incentive constraints are never violated given the optimal unconstrained allocation rule. 
Unfortunately in our model, such method fails since the optimal unconstrained allocation rule may not be implementable under IC contracts. 
We provide a detailed counterexample in \cref{sub:failure_of_implementation}.

By relaxing integral monotonicity to standard monotonicity and taking $u(\bar{c})=0$ we get that the maximum virtual welfare \emph{upper-bounds} the optimal expected revenue. Formally, let $\ironed$ denote the \emph{ironed} virtual cost function, achieved by standard ironing --- taking the convex hull of the integration of the virtual cost, and then taking the derivative~\citep{myerson1981optimal}.\footnote{Since the ironed virtual cost function may not be strictly monotone, we define its inverse as $\ironed^{-1}(q) := \sup\{c: \ironed(c) \leq q\}$.}
Using ironing~\citep{myerson1981optimal} and the optimality of maximizing virtual welfare with respect to all randomized allocation rules:


\begin{corollary}\label{cor:upper-bound}
The optimal expected revenue of a (randomized) IC contract is upper-bounded by%
\begin{align*}
\sup_{\text{\emph{monotone }}\alloc}\expect[c\sim G]{R_{\alloc(c)} - \gamma_{\alloc(c)} \virtual(c)} = \sup_{\alloc}\expect[c\sim G]{R_{\alloc(c)} - \gamma_{\alloc(c)} \ironed(c)}. 
\end{align*}
\end{corollary}

In the online appendix, we use these characterizations to show that optimal contracts may exhibit un-intuitive properties, such as a lack of revenue monotonicity in first-order stochastic dominance in the cost distribution.
This adds to the list of undesirable properties of optimal contracts in pure moral hazard 
settings, such as lack of interpretability and non-monotonicity of transfers in rewards, and provides further motivation for studying simple contracts such as linear contracts.

\section{Improved Approximation to Revenue Benchmark}
\label{apx:improved_revenue}

\subsection{Linearly-Bounded Distributions}
\label{apx:linear bound proof}
Our proof of \cref{thm:lin-bounded-const-apx}.\ref{item:with-highest-cost} relies on the characterization of monotone piecewise constant allocation rules in IC contracts. 
Note that a necessary property of the allocation rule that ensures IC is monotonicity. In our context, a \emph{monotone} allocation rule recommends actions that are more costly in terms of effort --- and hence more rewarding in expectation for the principal --- to agents whose cost per unit-of-effort is lower. 
Intuitively, such agents are better-suited to take on effort-intensive tasks. Since the number of actions is finite, the class of monotone allocation rules boils down to the class of \emph{monotone piecewise constant} rules. Informally, these are rules such that: (i) the allocation function $x(\cdot)$ is locally constant in distinct intervals within $[0,\bar{c}]$; and (ii) the allocation is decreasing with intervals. Formally:

\begin{definition}
\label{def:piecewise-constant}
    An allocation rule $x:C\to[n]$ is \emph{monotone piecewise constant} if there exist $\ell+2$ \emph{breakpoints} $\underline{c}=z_{\ell+1}< ...< z_{0}=\bar{c}$ where $\ell\leq n$, such that $\forall i \in [\ell]$: 
    \begin{enumerate}[topsep=0pt,itemsep=-1ex,partopsep=1ex,parsep=1ex]
        \item $c \in (z_{i+1},z_{i}] \implies x(c)=x(z_{i})$; 
        \item $x(z_{i+1})> x(z_{i})$. 
    \end{enumerate}
\end{definition}

\begin{proof}[Proof of \cref{thm:lin-bounded-const-apx}.\ref{item:with-highest-cost}] Assume that $i^*(r,\ironed(\bar{c}))=0$, and denote by $\vwel$ the optimal expected virtual welfare, and by $\apx$ the expected revenue of a linear contract with parameter $\alpha$; our goal is to show that $\apx \geq \frac{\eta(1-\alpha)}{1-\beta} \vwel$. 
Let $x_{\wel}$ be the welfare maximizing allocation rule, $x_{\vwel}$ be the \emph{virtual} welfare maximizing allocation rule, and let $x_{\apx}$ be the allocation rule induced by a linear contract with parameter~$\alpha$. That is, $\forall c\in C$
\begin{eqnarray}
x_{\wel}(c)&=&\arg\max_{i\in [n]}\{R_i-\gamma_i c \},\nonumber\\
x_{\vwel}(c)&=&\arg\max_{i\in [n]}\{R_i-\gamma_i \ironed (c)\},\label{eq:alloc}\\ 
x_{\apx}(c)&=&\arg\max_{i\in [n]}\{\alpha R_i-\gamma_i c\}\nonumber.
\end{eqnarray}
Denote the breakpoints (as defined in Definition~\ref{def:piecewise-constant}) of $x_{\wel},x_{\vwel}$ and $x_{\apx}$ by $\{z_{\wel,i}\}$, $\{z_{\vwel,i}\}$, and $\{z_{\apx,i}\}$ respectively. Let $|\{z_{\wel,i}\}|$ be the image size of $x_{\wel}$. Throughout the proof we will assume that $n+2=|\{z_{\wel,i}\}|$, i.e., that all actions are allocated by $x_{\wel}$. This is for ease of notation since in this case $x_{\wel}(z_{\wel,i})=i$, but the same arguments apply when $|\{z_{\wel,i}\}|<n+2$.
Since $\ironed(c),\alpha c$ are increasing and $x_{\wel}(\bar{c})=0$ since $i^*(r,\ironed(\bar{c}))=0$, one can observe that $x_{\vwel},x_{\wel}$ 
and $x_{\apx}$ have the same image size $n+2 =|\{z_{\vwel,i}\}|=|\{z_{\apx,i}\}|$.
Further, their breakpoints satisfy
\begin{eqnarray}
\label{eq:bp-alpha-beta}
\forall i\in [n+1]: &z_{\vwel,i} = \ironed^{-1}(z_{\wel,i}),& z_{\apx,i} = \alpha z_{\wel,i}.
\end{eqnarray}
To see why the second part of \eqref{eq:bp-alpha-beta} holds, consider for instance the breakpoint $z_{\apx,n}$ (where $\alpha R_n - \gamma_n z_{\apx,n} = \alpha R_{n-1} - \gamma_{n-1} z_{\apx,n}$), and the breakpoint $z_{\wel,n}$ (where $R_n - \gamma_n z_{\wel,n} = R_{n-1} - \gamma_{n-1} z_{\wel,n}$). Note that by reorganizing the above, $z_{\apx,n}$ can be rewritten as $\alpha(R_{n}-R_{n-1})/(\gamma_{n}-\gamma_{n-1})$. Similarly, $z_{\wel,n}$ can be rewritten as $(R_{n}-R_{n-1})/(\gamma_{n}-\gamma_{n-1})$. 

We first aim to lower bound $\apx$. 
Note that the expected revenue for action $i\in [n]$ in a linear contract with parameter $\alpha$ is $(1-\alpha) R_i$. Further, recall that all types $c\in (z_{\apx,i+1},z_{\apx,i}]$ are allocated with the same action of $x_{\apx}(z_{\apx,i})=i$. Thus, the expected revenue over all types can be phrased by the breakpoints as follows $\apx = \sum_{i\in [n]} (G(z_{\apx,i})-G(z_{\apx,i+1})) (1-\alpha) R_{i}$. By reorganizing the latter 
as shown in Claim~\ref{appx:claim-apx-sum} in Appendix~\ref{apx:linear bound proof} we have
\begin{eqnarray*}
\apx &=&\sum_{i\in [n-1]}G(z_{\apx,i+1})(1-\alpha)(R_{i+1}-R_{i}).
\end{eqnarray*}
Let $n-j = \min_{i\in [n+1]} \{z_{\vwel,i} \mid \kappa \leq z_{\vwel,i}\}$. In words, $z_{\vwel,n-j}$ is the first breakpoint that is larger than $\kappa$. To lower bound $\apx$ we show that $G(z_{{\apx},i+1}) \geq G(z_{{\vwel},i+1})$ $\forall i\in [n-j-1]$. Note that by \eqref{eq:bp-alpha-beta}, $z_{{\apx},i+1}=\alpha z_{{\wel},i+1}$. Further, by the condition that $\frac{1}{\alpha}c\leq \ironed(c)$ $\forall \kappa \leq c$, we have that $\ironed^{-1}(c) \leq \alpha c$ $\forall \kappa \leq c$ (since both sides of the inequality are increasing). Therefore, $\ironed^{-1}(z_{{\wel},i+1}) \leq \alpha z_{{\wel},i+1}= z_{{\apx},i+1}$ $\forall i \in [n-j-1]$.
Using~\eqref{eq:bp-alpha-beta} again for $z_{{\vwel},i+1}=\ironed^{-1}(z_{{\wel},i+1})$, we have that $z_{{\vwel},i+1} \leq z_{{\apx},i+1}$ $\forall i \in [n-j-1]$, which implies $G(z_{{\vwel},i+1}) \leq G(z_{{\apx},i+1})$ $\forall i \in [n-j-1]$ as desired. Therefore, and since $R_{i+1}\geq R_i$ $\forall i \in [n-1]$ we have that
\begin{eqnarray}\label{eq:alpha-bound}
\apx \geq \sum_{i\in [n-j-1]}G(z_{\vwel,i+1})(1-\alpha)(R_{i+1}-R_{i}).
\end{eqnarray}

Next, we lower bound the right hand side of \eqref{eq:alpha-bound} using $\vwel_[\kappa,\bar{c}]$. 
First, we show in Claim~\ref{appx:claim-opt-sum-kappa} in \cref{apx:linear bound proof} that the optimal virtual welfare can also be phrased using the breakpoints 
as follows: $\vwel_{[\kappa , \bar{c}]}=\sum_{i \in [n-j-1]}G(z_{\vwel,i+1})[R_{i+1}-R_{i}-z_{\vwel,i+1}(\gamma_{i+1}-\gamma_{i})]-G(\kappa )(R_{n-j}-\gamma_{x(\kappa )}\kappa ).$ 
Recall that by definition of $x_{\vwel}$, we have that $n-j\in \arg\max_{i\in [n]}\{R_i -\gamma_i \ironed(\kappa )\}$. Therefore, $R_i -\gamma_i \ironed(\kappa )\geq 0$. Which implies, since $\ironed(\kappa)\geq \kappa$ that $R_{n-j}-\gamma_{x(\kappa )}\kappa \geq 0$. Thus, 
\begin{eqnarray}\label{eq:opt-virtual}
\sum_{i \in [n-j-1]}G(z_{\vwel,i+1})[R_{i+1}-R_{i}-z_{\vwel,i+1}(\gamma_{i+1}-\gamma_{i})]&\geq & \vwel_{[\kappa , \bar{c}]}.
\end{eqnarray}
By Definition~\ref{def:piecewise-constant}, the expected welfare at the breakpoint $z_{\wel,i+1}$ given actions $i+1$ and $i$ are equal. 
Thus, $R_{i+1}-R_{i}=(\gamma_{i+1}-\gamma_{i}) z_{\wel,i+1}, \forall i \in [n-1]$. Multiplying both sides by $\beta$ and using $\beta z_{\wel,i+1} \leq \ironed^{-1}(z_{\wel,i+1})=z_{\vwel,i+1}$ $\forall i \in [n-j-1]$, 
we have that $\beta (R_{i+1}-R_{i})\leq (\gamma_{i+1}-\gamma_{i}) z_{\vwel,i+1}$ $\forall i \in [n-j-1]$.
Combining the latter with \eqref{eq:opt-virtual},
we have that 
\begin{eqnarray}\label{eq:beta-bound}
\sum_{i \in [n-j-1]}G(z_{\vwel,i+1})(1-\beta)(R_{i+1}-R_{i})&\geq & \vwel_{[\kappa , \bar{c}]}.
\end{eqnarray}
Combining \eqref{eq:alpha-bound}, \eqref{eq:beta-bound} and the fact that $\vwel_{[\kappa,\bar{c}]}\geq \eta\cdot \vwel$, we have that $\apx \geq \frac{\eta(1-\alpha)}{1-\beta}\vwel$. 
\end{proof}

The following claims completes the proof of Theorem~\ref{thm:lin-bounded-const-apx}.
\begin{claim}\label{appx:claim-apx-sum}
Consider $\apx$, and $z_{\apx,i}$ $\forall i\in [n+1]$ as defined in the proof of Theorem~\ref{thm:lin-bounded-const-apx}.
\begin{eqnarray*}
\apx &=&\sum_{i\in [n-1]}G(z_{\apx,i+1})(1-\alpha)(R_{i+1}-R_{i}).
\end{eqnarray*}
\end{claim}

\begin{proof} As in the proof of Theorem~\ref{thm:lin-bounded-const-apx}
\begin{eqnarray*}
\apx =
\sum_{i\in [n]} (G(z_{\apx,i})-G(z_{\apx,i+1})) (1-\alpha) R_{i}.
\end{eqnarray*}
Which, by splitting the sum is equal to
\begin{eqnarray*}
\sum_{i\in [n]} G(z_{\apx,i})(1-\alpha) R_{i}-\sum_{i\in [n]} G(z_{\apx,i+1}) (1-\alpha) R_{i}.
\end{eqnarray*}
Note that since $R_0=0$, the first term can be re-indexed as $\sum_{i\in [n-1]} G(z_{\apx,i+1})(1-\alpha) R_{i+1}$. Further, since $G(z_{\apx,n+1})=0$, we have that second term is equal to $\sum_{i\in [n-1]} G(z_{\apx,i+1})(1-\alpha) R_i.$ Thus, 
\begin{equation*}
\apx = \sum_{i\in [n-1]} G(z_{\apx,i+1})(1-\alpha) (R_{i+1}-R_{i}). \qedhere
\end{equation*}
\end{proof}

\begin{claim}\label{appx:claim-opt-sum-kappa}
Let $\kappa\in [\underline{c},\bar{c}]$, and consider $x_\vwel$, and $z_{\vwel,i}$ for any $i\in [n+1]$ as defined in the proof of Theorem~\ref{thm:lin-bounded-const-apx}. 
The welfare contribution $\vwel_{[\kappa , \bar{c}]}$ equals
\begin{eqnarray*}
\sum_{i \in [n-j-1]}G(z_{\vwel,i+1})[R_{i+1}-R_{i}-z_{\vwel,i+1}(\gamma_{i+1}-\gamma_{i})]-G(\kappa)(R_{n-j}-\gamma_{x(\kappa )}\kappa ),
\end{eqnarray*}
where $n-j = \min_{i\in [n+1]} \{z_{\vwel,i} \mid \kappa  \leq z_{\vwel,i}\}.$
\end{claim}

\begin{proof}
We will use $x$ instead of $x_{\vwel}$
and $z_i$ instead of $z_{\vwel,i}$ for simplicity.
We will show that the expected virtual cost over $[\kappa,\bar{c}]$ is at least
\begin{eqnarray*}
\sum_{i\in [n-j-1]}(G(z_{\vwel,i+1})-G(\kappa ))[z_{\vwel,i+1}(\gamma_{i+1}-\gamma_{i})].
\end{eqnarray*}
First, note that 
\begin{eqnarray*}
\int_{\kappa}^{\bar{c}}\gamma_{x(c)}g(c)\ironed(c) \dd c=\int_{\kappa}^{\bar{c}}\gamma_{x(c)}(cg(c)+G(c)) \dd c.
\end{eqnarray*}
Using integration by parts having $(g(c)c+G(c))$ as $(G(c)c)'$, the above is
\begin{eqnarray*}
G(c)c\gamma_{x(c)}\mid^{\bar{c}}_{\kappa}-\int_{\kappa}^{\bar{c}}\gamma'_{x(c)}G(c)c \dd c.
\end{eqnarray*}
Since $x(\bar{c})=0$ this is $-G(\kappa)\gamma_{x(\kappa)}\kappa-\int_{\kappa}^{\bar{c}}\gamma'_{x(c)}G(c)c \dd c$. Using $G(c)=G(\kappa )+\int^{c}_{\kappa}g(z)\dd z$, the latter is equal to
\begin{eqnarray*}
-G(\kappa)\gamma_{x(\kappa)}\kappa-G(\kappa)\int_{\kappa}^{\bar{c}}\gamma'_{x(c)}c \dd c-\int_{\kappa}^{\bar{c}}\gamma'_{x(c)}(\int^{c}_{\kappa}g(z)\dd z) c \dd c.
\end{eqnarray*}
Reversing the integration order, the above is equal to 
\begin{eqnarray}\label{eq:virtual-cost-1}
-G(\kappa)\gamma_{x(\kappa)}\kappa-G(\kappa)\int_{\kappa}^{\bar{c}}\gamma'_{x(c)}c \dd c-\int_{\kappa}^{\bar{c}}(\int^{\bar{c}}_{c}\gamma'_{x(z)} z\dd z) g(c) \dd c.
\end{eqnarray}
Note that 
\begin{eqnarray*}
\int^{c}_{\bar{c}} \gamma'_{x(z)} z \dd z=\sum_{k\in [i-1]} z_{\ell+1}(\gamma_{\ell+1}-\gamma_{\ell}),
\end{eqnarray*}
where $x(c)=i$. Therefore, and by the fact the $x(c)$ is constant on $(z_{i+1},z_{i}]$, we have that \eqref{eq:virtual-cost-1} is equal to
\begin{eqnarray*}
-G(\kappa)\gamma_{x(\kappa)}\kappa+G(\kappa)\sum_{k\in [n-j-1]} z_{\ell+1}(\gamma_{\ell+1}-\gamma_{\ell})&+&\\
\sum_{i \in [n-j-1]} (G(z_{i})-G(z_{i+1}))\sum_{\ell\in [i-1]} z_{\ell+1}(\gamma_{\ell+1}-\gamma_{\ell})&+&\\ (G(z_{n-j})-G(\kappa))\sum_{\ell\in [n-j-1]} z_{\ell+1}(\gamma_{\ell+1}-\gamma_{\ell}).
\end{eqnarray*}
Which is 
\begin{eqnarray*}
-G(\kappa)\gamma_{x(\kappa)}\kappa+G(z_{n-j})\sum_{\ell\in [n-j-1]} z_{\ell+1}(\gamma_{\ell+1}-\gamma_{\ell})&+&\\
\sum_{i \in [n-j-1]} (G(z_{i})-G(z_{i+1}))\sum_{\ell\in [i-1]} z_{\ell+1}(\gamma_{\ell+1}-\gamma_{\ell}) 
\end{eqnarray*}
Reversing the order of summation, the latter equals 
\begin{eqnarray*}
-G(\kappa)\gamma_{x(\kappa)}\kappa+G(z_{n-j})\sum_{\ell\in [n-j-1]} z_{\ell+1}(\gamma_{\ell+1}-\gamma_{\ell})&+&\\
\sum_{\ell\in [n-j-2]} z_{\ell+1}(\gamma_{\ell+1}-\gamma_{\ell}) \sum_{i=\ell+1}^{n-j-1} (G(z_{i})-G(z_{i+1}))
\end{eqnarray*}
Since $\sum_{i=\ell+1}^{n-j-1} (G(z_{i})-G(z_{i+1}))=G(z_{\ell+1})-G(z_{n-j})$, the above is
\begin{eqnarray*}
-G(\kappa)\gamma_{x(\kappa)}\kappa+G(z_{n-j})[z_{n-j}(\gamma_{n-j}-\gamma_{n-j-1})]+
\sum_{\ell\in [n-j-2]}G(z_{\ell+1}) z_{\ell+1}(\gamma_{\ell+1}-\gamma_{\ell}).
\end{eqnarray*}
That is, 
\begin{eqnarray}\label{eq:virtual-cost-final}
\sum_{\ell\in [n-j-1]}G(z_{\ell+1}) [z_{\ell+1}(\gamma_{\ell+1}-\gamma_{\ell})] -G(\kappa)\gamma_{x(\kappa)}\kappa.
\end{eqnarray}
By the same arguments, the expected reward is equal to 
\begin{eqnarray*}
\sum_{i \in [n-j-1]} (G(z_{i})-G(z_{i+1}))\sum_{k\in [i-1]}(R_{\ell+1}-R_{\ell})+ (G(z_{n-j})-G(\kappa))R_{n-j}.
\end{eqnarray*}
Which, by reversing the order of summation is equal to
\begin{eqnarray*}
\sum_{\ell \in [n-j-2]}(R_{\ell+1}-R_{\ell})\sum_{i=\ell+1}^{n-j-1}(G(z_{i})-G(z_{i+1}))+ (G(z_{n-j})-G(\kappa))R_{n-j}.
\end{eqnarray*}
That is,
\begin{eqnarray*}
\sum_{\ell \in [n-j-2]}(R_{\ell+1}-R_{\ell})[G(z_{\ell+1})-G(z_{n-j})]+ (G(z_{n-j})-G(\kappa))R_{n-j}.
\end{eqnarray*}
Which is, by splitting the first summation,
\begin{eqnarray*}
\sum_{\ell \in [n-j-2]}(R_{\ell+1}-R_{\ell})G(z_{\ell+1})-\sum_{\ell \in [n-j-2]}(R_{\ell+1}-R_{\ell})G(z_{n-j})+ (G(z_{n-j})-G(\kappa))R_{n-j}.
\end{eqnarray*}
The above is equal to
\begin{eqnarray*}
\sum_{\ell \in [n-j-2]}(R_{\ell+1}-R_{\ell})G(z_{\ell+1})-G(z_{n-j})R_{n-j-1}+ (G(z_{n-j})-G(\kappa))R_{n-j}.
\end{eqnarray*}
That is
\begin{eqnarray*}
\sum_{\ell \in [n-j-1]}G(z_{\ell+1})(R_{\ell+1}-R_{\ell})-G(\kappa)R_{n-j}.
\end{eqnarray*}
Combining the above with \eqref{eq:virtual-cost-final}, we have that the expected virtual welfare from types on $[\kappa,\bar{c}]$ is as follows.
\begin{equation*}
\sum_{\ell \in [n-j-1]}G(z_{\ell+1})[R_{\ell+1}-R_{\ell}-z_{\ell+1}(\gamma_{\ell+1}-\gamma_{\ell})]-G(\kappa)(R_{n-j}-\gamma_{x(\kappa)}\kappa). \qedhere
\end{equation*}
\end{proof}

\subsection{Tight Approximation for General Distributions
}\label{appx:apprx-n}

In this section we complete the proofs for the approximation results on general distributions. 
The proof of Theorem~\ref{thm:upper-bound-n} relies on the following claim.

\begin{claim}\label{appx:claim-opt-sum}
Consider $x_\vwel$, and $z_{\vwel,i}$ $\forall i\in [n+1]$ as defined in the proof of Theorem~\ref{thm:lin-bounded-const-apx}.
\begin{eqnarray}
\vwel &=&
\sum_{i\in [n-1]}G(z_{\vwel,i+1})[R_{i+1}-R_{i}-z_{\vwel,i+1}(\gamma_{i+1}-\gamma_{i})].
\end{eqnarray}
\end{claim}

\begin{proof}
We will use $x$ instead of $x_{\vwel}$
and $z_i$ instead of $z_{\vwel,i}$ for simplicity.
We will show that the expected virtual cost is equal to
\begin{eqnarray*}
\sum_{i\in [n-1]}G(z_{\vwel,i+1})[z_{\vwel,i+1}(\gamma_{i+1}-\gamma_{i})].
\end{eqnarray*}
First, note that 
\begin{eqnarray*}
\int_{\underline{c}}^{\bar{c}}\gamma_{x(c)}g(c)\ironed(c) \dd c=\int_{\underline{c}}^{\bar{c}}\gamma_{x(c)}(cg(c)+G(c)) \dd c.
\end{eqnarray*}
Using integration by parts having $(g(c)c+G(c))$ as $(G(c)c)'$, the latter is equal to
\begin{eqnarray*}
G(c)c\gamma_{x(c)}\mid^{\bar{c}}_{\underline{c}}-\int_{\underline{c}}^{\bar{c}}\gamma'_{x(c)}G(c)c \dd c
\end{eqnarray*}
Since $x(\bar{c})=0$ and $G(\underline{c})=0$, the first term equals zero
and hence the above equals $-\int_{\underline{c}}^{\bar{c}}\gamma'_{x(c)}G(c)c \dd c$. 
Using $G(c)=\int^{c}_{\underline{c}}g(z)\dd z$, the latter is $-\int_{\underline{c}}^{\bar{c}}(\int^{c}_{\underline{c}}g(z)\dd z) \gamma'_{x(c)} c \dd c$. Reversing the integration order, $-\int_{\underline{c}}^{\bar{c}}(\int_{c}^{\bar{c}} \gamma'_{x(z)} z \dd z) g(c) \dd c$. That is, 
\begin{eqnarray*}
\int_{\underline{c}}^{\bar{c}}\rbr{\int^{c}_{\bar{c}} \gamma'_{x(z)} z \dd z} g(c) \dd c
\end{eqnarray*}
Note that 
\begin{eqnarray*}
\int^{c}_{\bar{c}} \gamma'_{x(z)} z \dd z=\sum_{k\in [i]} z_{k+1}(\gamma_{k+1}-\gamma_{k}),
\end{eqnarray*}
where $x(c)=x(z_{i+1})$. Therefore, and by the fact the $x(c)$ is constant on $(z_i,z_{i+1}]$, we have that the expected virtual cost is equal to
\begin{eqnarray*}
\int_{\underline{c}}^{\bar{c}}\gamma_{x(c)}g(c)\ironed(c) \dd c=\sum_{i\in [n]}(G(z_{i+1})-G(z_i)) \sum_{k\in [i]} z_{k+1}(\gamma_{k+1}-\gamma_{k})
\end{eqnarray*}
Reversing the order of summation, the latter is equal to 
\begin{eqnarray*}
\sum_{k\in [n]}z_{k+1}(\gamma_{k+1}-\gamma_{k})\sum_{i\in [k]}(G(z_{i+1})-G(z_i))  
\end{eqnarray*}
Since $G(z_0)=0$, this is equal to $\sum_{k\in [n]}G(z_{k+1})[z_{k+1}(\gamma_{k+1}-\gamma_{k})]$. 
Replacing $i$ with $k$ we have that the virtual cost is equal to $\sum_{i\in [n]}G(z_{i+1})[z_{i+1}(\gamma_{i+1}-\gamma_{i})]$. By the same arguments, the expected reward equals
\begin{align*}
\sum_{i\in [n]}(G(z_{i+1})-G(z_i)) \sum_{k\in [i]} (R_{k+1}-R_{k})
&= \sum_{k\in [n]} (R_{k+1}-R_{k}) \sum_{i\in [k]}  (G(z_{i+1})-G(z_i)) \\
&= \sum_{k\in [n]} G(z_{k+1})(R_{k+1}-R_{k}),
\end{align*}
where the first equality holds by reversing the order of summation.
Therefore, the expected virtual welfare is $\sum_{i\in [n]} G(z_{i+1})[(R_{i+1}-R_{i})-z_{i+1}(\gamma_{i+1}-\gamma_{i})]$ as desired.
\end{proof}

\begin{proof}[Proof of Theorem~\ref{thm:upper-bound-n}]
Denote by $\vwel$ the optimal expected virtual welfare, and by $\apx$ the expected revenue of the optimal linear contract; our goal is to show that $\apx \geq \frac{1}{n} \vwel$. 

Let $\alpha \in [0,1]$. 
Similarly to the proof of Theorem~\ref{thm:lin-bounded-const-apx}, 
let $\{z_{\wel,i}\},\{z_{\vwel,i}\},$ and $\{z_{\apx,i}\}$ be the breakpoints for welfare maximizing, virtual welfare maximizing and linear contract with parameter $\alpha$ respectively. We will assume similarly that the number of breakpoints is $n+2$. Using the same arguments, we have that $\forall i\in [n+1]:$
\begin{eqnarray*}
z_{\vwel,i} = \ironed^{-1}(z_{\wel,i}),& z_{\apx,i} = \alpha z_{\wel,i}.
\end{eqnarray*}
From Claim~\ref{appx:claim-opt-sum} we have that
\begin{eqnarray*}
\vwel &=&
\sum_{i\in [n-1]}G(z_{\vwel,i+1})[R_{i+1}-R_{i}-z_{\vwel,i+1}(\gamma_{i+1}-\gamma_{i})].
\end{eqnarray*}
Thus, there exists a breakpoint $k+1\in [n+1]$ such that 
\begin{eqnarray}\label{eq:k-action-bound}
G(z_{\vwel,k+1})[R_{k+1}-R_{k}-z_{\vwel,k+1}(\gamma_{k+1}-\gamma_{k})]\geq \frac{1}{n}\vwel 
\end{eqnarray}
Take $\alpha\in [0,1]$ such that $\alpha z_{\wel,k+1} =\ironed^{-1}(z_{\wel,k+1})$.\footnote{Note that there exists such $\alpha$ since $c \leq \ironed(c)$ $\forall c\in C$, or equivalently $\ironed^{-1}(c) \leq c$.} Consider the revenue of a linear contract with parameter $\alpha.$ By Claim~\ref{appx:claim-apx-sum}, the revenue of this contract is 
\begin{eqnarray*}
\apx &=&\sum_{i\in [n-1]}G(z_{\apx,i+1})(1-\alpha)(R_{i+1}-R_{i}).
\end{eqnarray*}
Note that we can split the above sum as follows.
\begin{eqnarray}\label{eq:apx-n-sum}
\apx &=& G(z_{\apx,k+1})(1-\alpha)(R_{k+1}-R_{k})\nonumber\\
&+& \sum_{k \neq i\in [n-1]}G(z_{\apx,i+1})(1-\alpha)(R_{i+1}-R_{i}).
\end{eqnarray}
Further, note that since $z_{\vwel,k+1}=\ironed^{-1}(z_{\wel,k+1})=\alpha z_{\wel,k+1}$ by out choice of $\alpha$ we can replace $z_{\vwel,k+1}$ in \eqref{eq:k-action-bound} with $\alpha z_{\wel,k+1}$ to have
\begin{eqnarray*}
G(\alpha z_{\wel,k+1})[R_{k+1}-R_{k}-\alpha z_{\wel,k+1}(\gamma_{k+1}-\gamma_{k})]\geq \frac{1}{n}\vwel 
\end{eqnarray*}
As shown in the proof of Theorem~\ref{thm:lin-bounded-const-apx}, it follows from Definition~\ref{def:piecewise-constant} and the definition of $x_{\wel}(c)$ that 
$\alpha (R_{k+1}-R_{k})=\alpha\cdot z_{\wel,k+1}(\gamma_{k+1}-\gamma_{k})$. 
Combined with the inequality above, we have 
\begin{eqnarray*}
G(z_{\apx,k+1})(1-\alpha)[R_{k+1}-R_{k}]\geq \frac{1}{n}\vwel.
\end{eqnarray*}
Therefore, we can use the above to lower bound the first term in \eqref{eq:apx-n-sum}, and since the second term is non-negative we have that $\apx\geq \frac{1}{n}\vwel$ as desired.
\end{proof}

\bibliographystyle{apalike}
\bibliography{bibliography}

\begin{thebibliography}{}

\bibitem[Akbarpour et~al., 2023]{akbarpouralgorithmic}
Akbarpour, M., Kominers, S.~D., Li, K.~M., Li, S., and Milgrom, P. (2023).
\newblock Algorithmic mechanism design with investment.
\newblock {\em Econometrica}, 91(6):1969--2003.

\bibitem[Alon et~al., 2021]{AlonDT21}
Alon, T., D{\"{u}}tting, P., and Talgam{-}Cohen, I. (2021).
\newblock Contracts with private cost per unit-of-effort.
\newblock In {\em Proceedings of the {ACM} Conference on Economics and Computation ({EC})}, pages 52--69.

\bibitem[Babaioff et~al., 2012]{Babaioff2006combinatorial}
Babaioff, M., Feldman, M., Nisan, N., and Winter, E. (2012).
\newblock Combinatorial agency.
\newblock {\em Journal of Economic Theory}, 147(3):999--1034.

\bibitem[Babaioff et~al., 2017]{BabaioffGN17}
Babaioff, M., Gonczarowski, Y.~A., and Nisan, N. (2017).
\newblock The menu-size complexity of revenue approximation.
\newblock In {\em Proceedings of the {ACM} Symposium on Theory of Computing (STOC)}, pages 869--877.

\bibitem[Balamceda et~al., 2016]{BalamcedaEtAl16}
Balamceda, F., Balseiro, S., Correa, J., and Stier-Moses, N.~E. (2016).
\newblock Bounds on the welfare loss from moral hazard with limited liability.
\newblock {\em Games and Economic Behavior}, 95:137--155.

\bibitem[Bulow and Roberts, 1989]{bulow1989simple}
Bulow, J. and Roberts, J. (1989).
\newblock The simple economics of optimal auctions.
\newblock {\em Journal of Political Economy}, 97(5):1060--1090.

\bibitem[Carroll, 2015]{carol-robust-linear}
Carroll, G. (2015).
\newblock Robustness and linear contracts.
\newblock {\em American Economic Review}, 105(2):536--563.

\bibitem[Castiglioni et~al., 2025]{castiglioni2025reduction}
Castiglioni, M., Chen, J., Li, M., Xu, H., and Zuo, S. (2025).
\newblock A reduction from multi-parameter to single-parameter bayesian contract design.
\newblock In {\em Proceedings of the 2025 Annual ACM-SIAM Symposium on Discrete Algorithms (SODA)}, pages 1795--1836. SIAM.

\bibitem[Castiglioni et~al., 2021]{CastiglioniMG22}
Castiglioni, M., Marchesi, A., and Gatti, N. (2021).
\newblock Bayesian agency: Linear versus tractable contracts.
\newblock In {\em Proceedings of the ACM Conference on Economics and Computation (EC)}, pages 285--286.

\bibitem[Castiglioni et~al., 2022]{castiglioni2022designing}
Castiglioni, M., Marchesi, A., and Gatti, N. (2022).
\newblock Designing menus of contracts efficiently: The power of randomization.
\newblock In {\em Proceedings of the {ACM} Conference on Economics and Computation {(EC)}}, pages 705--735.

\bibitem[Castro-Pires et~al., 2024]{castro2024disentangling}
Castro-Pires, H., Chade, H., and Swinkels, J. (2024).
\newblock Disentangling moral hazard and adverse selection.
\newblock {\em American economic review}, 114(1):1--37.

\bibitem[Dai and Toikka, 2022]{DaiT22}
Dai, T. and Toikka, J. (2022).
\newblock Robust incentives for teams.
\newblock {\em Econometrica}, 90:1583--1613.

\bibitem[Devanur et~al., 2016]{devanur2016sample}
Devanur, N.~R., Huang, Z., and Psomas, C.-A. (2016).
\newblock The sample complexity of auctions with side information.
\newblock In {\em Proceedings of the ACM Symposium on Theory of Computing (STOC)}, pages 426--439.

\bibitem[Diamond, 1998]{Diamond98}
Diamond, P. (1998).
\newblock Managerial incentives: On the near linearity of optimal compensation.
\newblock {\em Journal of Political Economy}, 106(5):931--57.

\bibitem[D{\"{u}}tting et~al., 2021a]{DuttingEFK21}
D{\"{u}}tting, P., Ezra, T., Feldman, M., and Kesselheim, T. (2021a).
\newblock Combinatorial contracts.
\newblock In {\em Proceedings of the {IEEE} Symposium on Foundations of Computer Science {(FOCS)}}, pages 815--826.

\bibitem[D\"utting et~al., 2023]{DuettingEFK23}
D\"utting, P., Ezra, T., Feldman, M., and Kesselheim, T. (2023).
\newblock Multi-agent contracts.
\newblock In {\em Proceedings of the ACM Symposium on Theory of Computing (STOC)}, pages 1311--1324.

\bibitem[D{\"{u}}tting et~al., 2025]{DuttingEFK25}
D{\"{u}}tting, P., Ezra, T., Feldman, M., and Kesselheim, T. (2025).
\newblock Multi-agent combinatorial contracts.
\newblock In {\em Proceedings of the 2025 Annual {ACM-SIAM} Symposium on Discrete Algorithms ({SODA})}, pages 1857--1891.

\bibitem[D{\"u}tting et~al., 2024]{dutting2024ambiguous}
D{\"u}tting, P., Feldman, M., Peretz, D., and Samuelson, L. (2024).
\newblock Ambiguous contracts.
\newblock {\em Econometrica}, 92(6):1967--1992.

\bibitem[D{\"{u}}tting et~al., 2019]{dutting2020simple}
D{\"{u}}tting, P., Roughgarden, T., and Talgam{-}Cohen, I. (2019).
\newblock Simple versus optimal contracts.
\newblock In {\em Proceedings of the {ACM} Conference on Economics and Computation ({EC})}, pages 369--387.

\bibitem[D{\"{u}}tting et~al., 2021b]{DuttingRT21}
D{\"{u}}tting, P., Roughgarden, T., and Talgam{-}Cohen, I. (2021b).
\newblock The complexity of contracts.
\newblock {\em {SIAM} Journal on Computing}, 50(1):211--254.

\bibitem[Frick et~al., 2023]{FII-23}
Frick, M., Ryota, I., and Ishii, Y. (2023).
\newblock Monitoring with rich data.
\newblock Working paper.

\bibitem[Gottlieb and Moreira, 2022]{gottlieb2022simple}
Gottlieb, D. and Moreira, H. (2022).
\newblock Simple contracts with adverse selection and moral hazard.
\newblock {\em Theoretical Economics}, 17:1357--1401.

\bibitem[Grossman and Hart, 1983]{GrossmanHart83}
Grossman, S.~J. and Hart, O.~D. (1983).
\newblock An analysis of the principal-agent problem.
\newblock {\em Econometrica}, 51(1):7--45.

\bibitem[Guruganesh et~al., 2021]{guruganesh20}
Guruganesh, G., Schneider, J., and Wang, J. (2021).
\newblock Contracts under moral hazard and adverse selection.
\newblock In {\em Proceedings of the {ACM} Conference on Economics and Computation ({EC})}, pages 563--582.

\bibitem[Hart and Nisan, 2019]{menu-complexity}
Hart, S. and Nisan, N. (2019).
\newblock Selling multiple correlated goods: Revenue maximization and menu-size complexity.
\newblock {\em Journal of Economic Theory}, 183:991--1029.

\bibitem[Hart and Reny, 2015]{nonmonotonicty}
Hart, S. and Reny, P. (2015).
\newblock Maximal revenue with multiple goods: Nonmonotonicity and other observations.
\newblock {\em Theoretical Economics}, 10:893--922.

\bibitem[Hartline, 2012]{hartline2012approximation}
Hartline, J.~D. (2012).
\newblock Approximation in mechanism design.
\newblock {\em American Economic Review}, 102(3):330--336.

\bibitem[Hartline and Lucier, 2015]{hartline2015non}
Hartline, J.~D. and Lucier, B. (2015).
\newblock Non-optimal mechanism design.
\newblock {\em American Economic Review}, 105(10):3102--3124.

\bibitem[Holmstr\"om, 1979]{Holmstrom79}
Holmstr\"om, B. (1979).
\newblock Moral hazard and observability.
\newblock {\em The Bell Journal of Economics}, 10(1):74--91.

\bibitem[Holmstr\"om and Milgrom, 1987]{milgrom-holmstrom87}
Holmstr\"om, B. and Milgrom, P. (1987).
\newblock Aggregation and linearity in the provision of intertemporal incentives.
\newblock {\em Econometrica}, 55:303--328.

\bibitem[Innes, 1990]{innes1990limited}
Innes, R.~D. (1990).
\newblock Limited liability and incentive contracting with ex-ante action choices.
\newblock {\em Journal of Economic Theory}, 52(1):45--67.

\bibitem[Kambhampati, 2023]{kambhampati2023randomization}
Kambhampati, A. (2023).
\newblock Randomization is optimal in the robust principal-agent problem.
\newblock {\em Journal of Economic Theory}, 207:105585.

\bibitem[Kambhampati et~al., 2025]{kambhampati2025randomization}
Kambhampati, A., Peng, B., Tang, Z.~G., Toikka, J., and Vohra, R. (2025).
\newblock Randomization and the robustness of linear contracts.
\newblock {\em Working paper}.

\bibitem[Li, 2022]{li2021selling}
Li, Y. (2022).
\newblock Selling data to an agent with endogenous information.
\newblock In {\em Proceedings of the {ACM} Conference on Economics and Computation {(EC)}}, pages 664--665.

\bibitem[Milgrom and Segal, 2002]{MilgromS02}
Milgrom, P. and Segal, I. (2002).
\newblock Envelope theorems for arbitrary choice sets.
\newblock {\em Econometrica}, 70(2):583--601.

\bibitem[Myerson, 1981]{myerson1981optimal}
Myerson, R.~B. (1981).
\newblock Optimal auction design.
\newblock {\em Mathematics of Operations Research}, 6(1):58--73.

\bibitem[Myerson, 1982]{Myerson82}
Myerson, R.~B. (1982).
\newblock Optimal coordination mechanisms in generalized principal-agent problems.
\newblock {\em Journal of Mathematical Economics}, 10:67--81.

\bibitem[Roughgarden and Talgam-Cohen, 2019]{roughgarden2019approximately}
Roughgarden, T. and Talgam-Cohen, I. (2019).
\newblock Approximately optimal mechanism design.
\newblock {\em Annual Review of Economics}, 11:355--381.

\bibitem[Spielman and Teng, 2004]{spielman2004smoothed}
Spielman, D.~A. and Teng, S.-H. (2004).
\newblock Smoothed analysis of algorithms: Why the simplex algorithm usually takes polynomial time.
\newblock {\em Journal of the ACM}, 51(3):385--463.

\bibitem[Walton and Carroll, 2022]{WaltonC22}
Walton, D. and Carroll, G. (2022).
\newblock A general framework for robust contracting models.
\newblock {\em Econometrica}, 90(5):2129--2159.

\bibitem[Yang, 2022]{yang2022selling}
Yang, K.~H. (2022).
\newblock Selling consumer data for profit: Optimal market-segmentation design and its consequences.
\newblock {\em American Economic Review}, 112(4):1364--93.

\bibitem[Yu and Kong, 2020]{YuK20}
Yu, Y. and Kong, X. (2020).
\newblock Robust contract designs: Linear contracts and moral hazard.
\newblock {\em Operations Research}, 68(5):1457--1473.

\end{thebibliography}

\newpage
\begin{center}
\large
\MakeUppercase{Online Appendix}
\end{center}
\section{Instantiations of Approximation Guarantees}
\label{appx:implications}

We provide proofs for \cref{cor:wel-implications} and \cref{cor:rev-implications}.

\begin{proof}[Proof of Corollary~\ref{cor:wel-implications}] For any $\kappa>1$, let $\alpha =\frac{\kappa+1}{2\kappa}$ and $\beta =\frac{1}{2}$. 
We show that $(\alpha,\beta,\kappa\underline{c})$-slowly increasing condition is satisfied. 
Note that for any $c\geq \kappa\underline{c}$, 
since $(2\alpha -1)c= (\frac{\kappa+1}{\kappa}-1)c\geq \underline{c}$, we have that $\alpha c-\underline{c} \geq c-\alpha c$, 
which further implies that the total probability in $[\underline{c},\alpha c]$ is larger than that in $[\alpha c,c]$
since $g(c)$ is non-increasing. 
Therefore, we have 
\begin{eqnarray*}
G(\alpha c)=\int^{\alpha c}_{\underline{c}}g(z)\dd z\geq \int^{c}_{\alpha c}g(z)\dd z=G(c)-G(\alpha c) & \forall c \geq \kappa.
\end{eqnarray*}
By adding $G(\alpha c)$ to both sides of the above inequality, we have $\beta=\frac{1}{2}$.
By \cref{thm:slow}, we have an approximation guarantee of $\frac{1}{(1-\frac{\kappa+1}{2\kappa})\frac{1}{2}\eta}=\frac{4\kappa}{(\kappa-1)\eta}$ as desired.
\end{proof}

\begin{proof}[Proof of Corollary~\ref{cor:rev-implications}.\ref{item:rev-non-increasing-dense}] We show that $(\frac{\kappa}{2\kappa-1},0,\kappa\underline{c})$-linear boundedness is satisfied. 
Note that  $\varphi(c)=\int^c_{\underline{c}}\varphi'(z)\dd z + \virtual(\underline{c})$,
$\virtual(\underline{c}) = \underline{c}$,
and $\varphi'(z)=2-\frac{G(z)g'(z)}{g^2(z)}\geq 2$ because of decreasing density. 
Therefore, the virtual cost $\varphi(c)\geq 2(c-\underline{c})+\underline{c}=2c-\underline{c}$. 
For any $c\geq \kappa\underline{c}$, it can be verified that the latter is at least $\frac{2\kappa -1}{\kappa}\cdot c$. 
By \cref{thm:lin-bounded-const-apx}, we have an approximation guarantee of $\frac{1}{\eta (1-\frac{\kappa}{2\kappa -1})} = \frac{2\kappa -1}{\eta (\kappa -1)}$ as desired.
\end{proof}

\begin{proof}[Proof of Corollary~\ref{cor:rev-implications}.\ref{item:rev-uniform}]
The virtual cost for uniform distribution $U[0,\bar{c}]$ is $\varphi(c)=2c.$ Corollary~\ref{cor:rev-implications}.\ref{item:rev-uniform} is therefore implied by Theorem~\ref{thm:lin-bounded-const-apx}. 
\end{proof}

\begin{proof}[Proof of Corollary~\ref{cor:rev-implications}.\ref{item:rev-normal}]
Consider the normal distribution $\mathcal{N}(\mu,\sigma^2)$ defined by the following CDF and PDF respectively.
\begin{eqnarray*}
G(c)=\frac{1}{2}[1+\erf(\frac{c-\mu}{\sigma \sqrt{2}})], & g(c)=\frac{1}{\sigma \sqrt{2\pi}} e^{-\frac{1}{2}(\frac{c-\mu}{\sigma})^2} & \forall c\geq 0,
\end{eqnarray*}
where $\erf(c)$ is the Gauss error function. By definition of the virtual cost, $\varphi(c)$ is as follows.\footnote{Note that this distribution has a point mass at type $c=0$. }
\begin{eqnarray}\label{eq:virtual-cost-normal}
\varphi(c)=c+\frac{\frac{1}{2}[1+\erf(\frac{c-\mu}{\sigma \sqrt{2}})]}{\frac{1}{\sigma \sqrt{2\pi}} e^{-\frac{1}{2}(\frac{c-\mu}{\sigma})^2}}=c+\frac{\sigma\sqrt{\pi}}{\sqrt{2}}[1+\erf(\frac{c-\mu}{\sigma \sqrt{2}})]e^{\frac{1}{2}(\frac{c-\mu}{\sigma})^2}.
\end{eqnarray}
Note that since $c\geq 0$, and by the assumption that $\frac{2\sqrt{2}}{5}\sigma\geq \mu$ it holds that $\frac{c-\mu}{\sigma \sqrt{2}}\geq -\frac{2}{5}$. Therefore, and since $\erf(\cdot)$ is increasing $\erf({\frac{c-\mu}{\sigma \sqrt{2}}})\geq\erf(-\frac{2}{5}) \geq -\frac{1}{2}$. Combining this with~\eqref{eq:virtual-cost-normal}, we have that 
\begin{eqnarray*}
\varphi(c)\geq c+\frac{1}{2}\frac{\sigma\sqrt{\pi}}{\sqrt{2}}e^{\frac{1}{2}(\frac{c-\mu}{\sigma})^2}\geq c+\frac{\sigma\sqrt{\pi}}{2\sqrt{2}}(1+\frac{1}{2}(\frac{c-\mu}{\sigma})^2),
\end{eqnarray*}
where the last inequality follows from the fact that $e^x\geq x+1$ $\forall x\geq 0$. Rearranging the above, we have that
\begin{eqnarray*}
\varphi(c)\geq c+\frac{\sigma\sqrt{\pi}}{2\sqrt{2}}(1+\frac{c^2-2\mu c+\mu^2}{2\sigma^2}),
\end{eqnarray*}
To show that $(\frac{2}{3},0,0)$-linear boundedness is satisfied we need to show that $\varphi(c)\geq \frac{3}{2}c.$ Using the above, it suffices to show that \begin{eqnarray*}
\frac{\sigma\sqrt{\pi}}{2\sqrt{2}}(1+\frac{c^2-2\mu c+\mu^2}{2\sigma^2})\geq \frac{1}{2}c.
\end{eqnarray*}
By rearranging, this holds if and only if $\sigma\sqrt{\pi}+ \sigma\sqrt{\pi}(\frac{c^2-2\mu c+\mu^2}{2\sigma^2})\geq \sqrt{2} c.$ Which is equal to $\sigma\sqrt{\pi}+ \sqrt{\pi}(\frac{c^2-2\mu c+\mu^2}{2\sigma})\geq \sqrt{2} c.$ The latter holds if and only if $2\sigma^2\sqrt{\pi}+ \sqrt{\pi}{(c^2-2\mu c+\mu^2)}\geq 2\sqrt{2} c\sigma$. That is, $\sqrt{\pi} c^2 -(2\mu \sqrt{\pi}+2\sqrt{2} \sigma )c+\sqrt{\pi}\mu^2
+2\sigma^2\sqrt{\pi} \geq 0.$ Note that the left hand side can be viewed as a function of $c$. This function has no roots if $(2\mu \sqrt{\pi}+2\sqrt{2} \sigma )^2-4\sqrt{\pi}(\sqrt{\pi}\mu^2
+2\sigma^2\sqrt{\pi})\leq 0$. Which holds if and only if $\mu^2 \pi+2 \mu \sqrt{\pi} \sqrt{2} \sigma+ 2 \sigma^2 \leq \pi \mu^2
+2\sigma^2 \pi.$ Rearranging the latter we have $2 \mu \sqrt{\pi} \sqrt{2} \sigma+ 2 \sigma^2 \leq
2\sigma^2 \pi$. Dividing by $2\sigma$, this is equal to $\mu \sqrt{2\pi} +  \sigma \leq 
\sigma \pi$. Which holds if and only if $\mu \frac{\sqrt{2\pi}}{(\pi -1)}  \leq 
\sigma $. Since $\frac{\sqrt{2\pi}}{(\pi -1)}\leq \frac{5}{2\sqrt{2}}$ and by the assumption that $\mu \frac{5}{2\sqrt{2}} \leq 
\sigma$ we have that $\varphi(c)\geq \frac{3}{2}c$ and the proof is complete.
\end{proof}

\section{Smoothed Analysis of Linear Contracts}
\label{appx:smoothed}
 
We give here an alternative perspective of our main result through the smoothed-analysis framework: The result that linear contracts work well for principal-agent instances in which the agent’s type distribution is not point-mass-like can also be viewed as stating that linear contracts work well for ``smoothed’’ distributions. 

In particular, adopting the same technique of slowly-increasing CDFs, 
we develop a smoothed analysis approach for contract design. 
We show that constant approximations with linear contracts can be obtained by slightly perturbing any known agent type (equivalently, any point mass distribution) with a uniform noise. 
The exact approximation ratio will be parameterized by the level of noise $\epsilon$ 
in the perturbed distribution. 
To simplify the exposition, we normalize the instance such that the known type is $1$.

\begin{proposition}[Smoothed analysis] 
\label{prop:smooth}
For any $\epsilon\in(0,1)$, consider any distribution $G$ such that with probability $1-\epsilon$, it is a point mass at $1$, 
and with probability $\epsilon$ it is drawn from the uniform distribution on $[0,2]$. 
A linear contracts achieves a $\frac{4(2-\epsilon)}{\epsilon}$-approximation to the optimal welfare for distribution~$G$.
\end{proposition}

\begin{proof}
Note that any principal-agent instance with lowest support $\underline{c}=0$ for type distribution $G$ has $(0,1)$-thin-tail.
Moreover, for any distribution $G$ that is generated by adding uniform noise as in the statement of \cref{prop:smooth} 
is $(\frac{1}{2}, \frac{\epsilon}{2(2-\epsilon)}, 0)$-slow-increasing. 
By applying \cref{thm:slow}, 
the approximation ratio of linear contracts to optimal welfare is $\frac{4(2-\epsilon)}{\epsilon}$.
\end{proof}

\section{Undesirable Properties of Optimal Contracts}
\label{sec:undesriable-features-of-opt-contracts}

We next explore the properties of optimal contracts in our single-dimensional private types setting. We focus on deterministic contracts as introduced in \cref{sec:contracts}.
We first show that the optimal principal utility may exhibit unnatural non-monotonicity. 
Specifically, the optimal revenue of the principal may be strictly smaller when contracting with a stronger agent (in terms of lower cost per unit-of-effort).
Secondly, we show that the amount of communication required to describe the optimal contract to the agent can be unbounded. 
These two properties make it less desirable to implement the optimal contract in practical situations.\footnote{Arguments similar to the ones developed in this section apply to randomized contracts, and lead to qualitatively similar results. An additional disadvantage of randomized contracts is their increased complexity relative to deterministic contracts.}

\subsection{Non-monotonicity of the Optimal Revenue}
\label{sub:non-monotone}
One way in which optimal contracts defy intuition is that a principal contracting with a stronger agent cannot necessarily expect higher revenue. 
Such non-monotonicity is known for auctions but only for \emph{multi}-dimensional types~\citep{nonmonotonicty}.\footnote{The multi-dimensionality in types is necessary for the non-monotonicity results in auction settings. 
Indeed, \citet{devanur2016sample} show that strong revenue monotonicity holds for single-dimensional auction settings.
}
To show this for contracts with \emph{single}-dimensional types, we keep the contract setting fixed and consider two type distributions $G$ and $H$, where $H$ first-order stochastically dominates $G$. Formally, $H \succ_{\rm FOSD} G$ if $G(c) \geq H(c)$ $\forall c\in C$, that is, types drawn from $G$ are more likely to have lower cost than those drawn from $H$. One would expect that the principal's revenue when agent costs are distributed according to $G$ would be higher, but this is not always the case --- the proof relies on \cref{ex:non-monotone}. 

\begin{proposition}\label{prop:fosd}
There exist a principal-agent setting and distributions $H \succ_{\rm FOSD} G$ such that the optimal expected revenue from an agent drawn from $G$ is smaller than an agent drawn from $H$.
\end{proposition}

\begin{example}\label{ex:non-monotone}
Let $\delta, \epsilon \to 0$, and consider the following parameterized setting with $4$ actions and $4$ outcomes.  
The costs in units-of-effort per action are $\gamma_0=0,\gamma_1=0,\gamma_2=\frac{1}{2},\gamma_3=\frac{1}{\delta}$, and the rewards are $r_0=r_1=r_2=0$ and $r_3=\frac{1}{\delta}$. The distributions over outcomes per action are $F_{0}=(1,0,0,0)$, $F_{1}=(0,1-0.25\delta,0,0.25\delta)$, $F_{2}=(0,0.5-\delta,0.5,\delta)$, $F_{3}=(0,0,1-  \delta-\delta^2 , \delta +\delta^2)$.
The type distributions $H \succ_{\rm FOSD} G$ are of the form:%
\footnote{For simplicity $H$ is a point mass distribution and does not satisfy 
the assumption that $\underline{c}<\bar{c}$, but it can be slightly perturbed to satisfy it.} 
\begin{eqnarray*}
G(c) = \begin{cases}
    \epsilon & c\in [0,1),\\
    1-\epsilon & c=1,
    \end{cases} 
    &&
H(c) = \begin{cases}
    0 & c\in [0,1),\\
    1 & c=1.
    \end{cases}
\end{eqnarray*}
\end{example}

\noindent
We outline the proof intuition here and defer the details later. The three non-null actions have the following characteristics. Action $1$ requires no effort and generates small expected reward, action $2$ requires little effort and generates moderate reward, and action $3$ requires very high effort and generates slightly higher reward than action 2.
Consider type $c=1$; incentivizing action $2$ is optimal for the principal. But the construction is such that incentivizing action $2$ can only be achieved via payment for outcome~$2$, which in turn implies extremely large payment for action $3$. Unlike type $c=1$ who never takes action~$3$ due to its high cost, type $c=0$ has the same zero cost for all actions. If the type distribution is such that $c=0$ can occur (i.e., $G$ rather than $H$), the principal will more often have to transfer this large payment (for outcome $2$) to the agent. There is thus a trade-off from having a low-cost type in the support, which ultimately leads to loss of revenue.

\begin{claim}
Consider Example~\ref{ex:non-monotone} and take $\epsilon$ and $\delta$ such that $\epsilon<\delta$, $(1+\delta)\delta<0.25$, and $7\delta +6\epsilon< 0.5$. The optimal expected revenue given $H$ is $\ge \frac{1}{2}$, while the optimal expected revenue given $G$ is $<\frac{1}{2}$.
\end{claim}

\begin{proof}
We denote the optimal expected revenue of the distributions $H$ and $G$ by $\opt(H)$ and $\opt(G)$ respectively. Consider $H$, this is a single principal single agent case. By setting the contract $t_0=t_1=t_3=0,t_2=1$, we have that the agent chooses action $2$, yielding an expected reward of $1$ and expected payment of $\frac{1}{2}$. Therefore, $\opt(H)\geq \frac{1}{2}$. Consider $G$, the following case analysis shows that the revenue from any contract $t$ is strictly less than $\frac{1}{2}$. 
\begin{enumerate}
    \item First, consider the case where $t$ incentivizes type $c=1$ to take action $3$. Since to do so the principal's payment for $c=1$ is at least $\frac{1}{\delta}$ (the agent's cost) and the expected reward from this action is at most $1+\delta$, the principal's utility from $c=1$ is at most $1+\delta-\frac{1}{\delta}$. Note that the principal's utility from $c=0$ is at most $1+\delta$ (the highest expected reward) so the expected utility is bounded by $2(1+\delta)-\frac{1}{\delta}$ which is strictly negative since $(1+\delta)\delta<\frac{1}{2}$. 
    \item The second case is when $t$ incentivizes type $c=1$ to take action $2$. We first show that for the principal to extract a payoff of at least $\frac{1}{2}$ using $t$, it must hold that the expected payment for action $2$, i.e., $(\frac{1}{2}-\delta)t_1+ \frac{1}{2}t_2+\delta t_3$ is upper bounded by $\frac{1}{2}+\epsilon$. To see this, note that if type $c=1$ takes action $2$, the expected payment is at least $\frac{1}{2}$ (the cost of action $2$). Since $c=0$ takes the action with highest expected payment, the principal's payoff from $c=0$ is bounded by $\frac{1}{2}+\delta$. Assume that the payment for type $c=1$ is more than $\frac{1}{2}+\epsilon$, the expected payoff is at most $(1-\epsilon) (\frac{1}{2}-\epsilon)+\epsilon (\frac{1}{2}+\delta)= \frac{1}{2} -\epsilon(1-\epsilon)  +\epsilon \delta = \frac{1}{2} - \epsilon (1-\epsilon -\delta)$, which is strictly less than $\frac{1}{2}$ since $\epsilon+\delta < 1.$  Therefore, we have that $(\frac{1}{2}-\delta)t_1+ \frac{1}{2}t_2+\delta t_3   \leq \frac{1}{2}+\epsilon$, which implies (since $\delta<\frac{1}{2}$ and $0 \leq t_1$) 
    \begin{equation}\label{eq:non-mon-1}
    \frac{1}{2}t_2+\delta t_3   \leq \frac{1}{2}+\epsilon.
    \end{equation}
    Also note that in order to incentivize type $c=1$ to take action $2$ it must hold that the agent's payoff from action $2$ is at most as his payoff from action $1$, $(1-0.25\delta)t_1 + 0.25\delta t_3 \leq 0.5 t_2 + \delta t_3 - 0.5$, i.e., $0.5+(1-0.25\delta)t_1  \leq 0.5 t_2 + 0.75 \delta t_3$, and since $0 \leq t_1$ it holds that 
    \begin{equation}\label{eq:non-mon-2}
    \frac{1}{2} \leq \frac{1}{2}t_2+\frac{3}{4}\delta t_3.
    \end{equation}
    Combining \eqref{eq:non-mon-1} and \eqref{eq:non-mon-2}, we have that $\frac{1}{2}t_2+\delta t_3 -\epsilon \leq \frac{1}{2} \leq \frac{1}{2}t_2+\frac{3}{4}\delta t_3,$ which implies $\frac{1}{4}\delta t_3  \leq \epsilon,$ i.e., $t_3  \leq \frac{4 \epsilon}{\delta }.$ Using this back in \eqref{eq:non-mon-2}, we have that $1 -6 \epsilon  \leq t_2.$ This implies that the expected payment for action $3$ is at least $(1-\delta -\delta^2) (1-6\epsilon) = 1 -6\epsilon - \delta (1+\delta) + 6\epsilon(\delta +\delta^2)> 1-6(\delta+\epsilon)$. Which is strictly more that $\frac{1}{2}+\delta$ since $\frac{1}{2}>7\delta + 6\epsilon$. This implies that the principal's payoff from type $c=0$ is strictly less than $\frac{1}{2}$. Since the optimal payoff from type $c=1$ is $\frac{1}{2}$ we have that the expected payoff from both types is strictly less than $\frac{1}{2}$.
    \item The third case is when $t$ incentivizes type $c=1$ to take action $1$. The expected reward is thus bounded by $(1+\delta)\epsilon+0.25 (1-\epsilon) = 0.25+(0.75+\delta)\epsilon< 0.25+(1+\delta)\epsilon<\frac{1}{2}.$
    \item The final case is when $t$ incentivizes type $c=1$ to take action $0$. The expected reward is at most $(1+\delta)\epsilon$ which is strictly less than $\frac{1}{2}$ since  $\epsilon<\delta,$ and $\delta (1+\delta)<\frac{1}{2}$.\qedhere
\end{enumerate}
\end{proof}

\subsection{Menu-size Complexity}
\label{sub:menu_complexity}
By the taxation principle, every IC contract has an equivalent non-direct-revelation representation, 
which offers the agent a \emph{menu} of payment profiles (rather than asking the agent to reveal his type then presenting him with a payment profile). 
Given a menu of payment profiles,
the agent will choose a payment profile from the menu and an 
action to maximize his expected utility based on his private type. 

In this section we adapt the measure of \emph{menu-size complexity} suggested for auctions by \citet{menu-complexity}, to our context of contracts with private types. 
Formally, the menu-size complexity of a contract design instance is the number of different payment profiles offered in the menu representation. 
\begin{definition}
\label{def:menu}
The \emph{menu-size complexity} of a payment rule $t:C\to \mathbb{R}^{m}$ is the image size of function~$t$.\footnote{In this paper, we only capture the menu size incurred by offering different payment schemes to the agents, reflecting the complexity due to screening. The additional complexity from recommending different actions to different types, even under the same payment scheme, is not included in our definition of menu-size complexity.
}
\end{definition}

The menu-size complexity captures the amount of information transferred from the principal to the agent in a straightforward communication of the contract. In auctions it has been formally tied to communication complexity~\cite{BabaioffGN17}. 
Intuitively, a contract with larger menu-size complexity requires more complex communication to describe the contract to the agent, 
and hence harder to be implemented in practice. 
We show that even with constant number of possible outcomes, 
the menu-size complexity of the optimal contract can be arbitrarily large. 
In contract, linear contracts have menu-size complexity of~1. 
 
\begin{proposition}
\label{pro:menu}
For every $n>1$ there exists a principal-agent setting with $m=2$ outcomes such that the menu-size complexity of the optimal contract is at least $\frac{n-1}{2}$.
\end{proposition}

The proof relies on the following example.

\begin{example}\label{ex:menu}
There is a set of actions $[n]$, with required effort levels $\gamma_i=\frac{i^2}{n}$ $\forall i\in [n]$. There are three outcomes $0=r_0<r_1<r_2$ such that $r_1+2(n-1)+1<r_2$. The distributions over outcomes are given by $F_0=(1,0,0)$, $F_i=(0,1-\frac{i}{n},\frac{i}{n})$ $\forall i>0$. The distribution over types is $G(c)=U[1,\Bar{c}]$, where $\frac{r_2+(n-1)r_1}{2}+\frac{1}{2}<\Bar{c}$. In this example, the menu-size complexity of the optimal contract is at least $\frac{n-1}{2}$.
\end{example}

Before analyzing Example~\ref{ex:menu} to formally prove Proposition~\ref{pro:menu}, we give a high-level intuition: The characterization in Lemma~\ref{lem:primal characterization} shows that every allocation~$x(c)$ requires a \emph{unique} expected payment $T^c_{x(c)}$. By incentive compatibility, this implies that all types with the same allocation have the same payment. Therefore, every action $i$ can be associated with payment $T_i$ required to incentivize $x$. In the pure adverse selection model, i.e., $F_{i,i}=1$ for all $i$, the single payment scheme $(T_i)_{i\in [n]}$ results in these payments. 
In our model with hidden action, however, the linear dependency of the actions' distributions introduces new constraints on the payment schemes and might increase the menu complexity.

To show that the menu size complexity of the optimal contract in Example~\ref{ex:menu} is at least $\frac{n-1}{2}$ we rely on the following claims.

\begin{claim}
\label{cla:menu-allocation-rule}
Consider Example~\ref{ex:menu}, and let $\Delta_r =r_2-r_1$. The virtual welfare maximizing allocation rule $x^*$ is defined by the following breakpoints $z_0=1, z_i=\frac{\Delta_r }{4(n-i)+2}+\frac{1}{2}, z_n=\frac{n r_1 +\Delta_r +1 }{2},z_{n+1} =\Bar{c}$ $\forall 1 \leq i < n$.
\end{claim}

\begin{proof}
To see this, note that the virtual cost is $\varphi(c)=2c-1$ and therefore the virtual welfare of action $i$ is given by $(1-\frac{i}{n})r_1+\frac{i}{n}r_2-\gamma_i \varphi(c)$ which is $r_1+ \frac{i}{n}\Delta_r -\frac{i^2}{n}(2c-1)$. It thus holds that action $i+1$ has higher virtual welfare than action $i$ if and only if 
\begin{eqnarray*}
r_1+ \frac{i+1}{n}\Delta_r -\frac{(i+1)^2}{n}(2c-1) \geq r_1+ \frac{i}{n}\Delta_r -\frac{i^2}{n}(2c-1) &\iff&\\
\frac{1}{n}\Delta_r  \geq   \frac{(i+1)^2}{n}(2c-1)-\frac{i^2}{n}(2c-1)=\frac{(2i+1)}{n}(2c-1)
\end{eqnarray*}
Equivalently, if and only if $\frac{\Delta_r}{4i+2} +\frac{1}{2}\geq c$. Therefore, action $n$ maximizes virtual welfare if and only if $1\leq c\leq \frac{\Delta_r}{4(n-1)+2}+\frac{1}{2}$ (not that $r_1+2(n-1)+2<r_2$), action $1\leq i <n$ maximizes virtual welfare if and only if $\frac{\Delta_r}{4i+2} +\frac{1}{2} \leq c\leq \frac{\Delta_r}{4(i-1)+2} +\frac{1}{2}$.
Furthermore, action $1$, has higher welfare than action $0$ when $0 \leq (1-\frac{1}{n})r_1+\frac{1}{n} r_2 -\frac{1}{n}(2c-1)$. That is, $c\leq \frac{n r_1+ \Delta_r +1}{2}$. So action $0$ maximizes virtual welfare given that $c\leq \frac{n r_1+ \Delta_r +1}{2}$.
\end{proof}

\begin{claim}
\label{cla:menu-UB}
The virtual welfare maximizing allocation rule $x^*$ is implementable by the following contract $t$ that has menu-size complexity of $\frac{n-1}{2}$, and for which $T^{\Bar{c}}_{x(\Bar{c})}=0$. 
\begin{align*}
t(c) &=
\rbr{0,\frac{r_1}{2}+\frac{2k-4k^2}{2n},\frac{r_1}{2}+\frac{2k-4k^2}{2n}+\frac{\Delta_r +4k-1}{2}}, \\
x^*(c)&=i\in \{ 2k-1,2k\}, \quad k\in \sbr{1,\frac{n}{2}},
\end{align*}
and for $c$ such that $x^*(c)=0$ we can use any of the above contracts.
\end{claim}

\begin{proof}
We first show that the contract that gives the highest expected payment for action $i$ is $t(c)$ for $x^*(c)=i$.
To do so, we show that given action $i$, the $k$ in the above contract that maximizes the expected payment for $i$ is $i=2k$ when $i$ is even, and $i=2k-1$ when $i$ is odd. First, the expected payment for action $i$ given some $k$ is as follows.
\begin{align*}
&(1-\frac{i}{n}) t(c)_1 + \frac{i}{n} t(c)_2 
= t(c)_1 +\frac{i}{n} (t(c)_2 - t(c)_1) \\
& = \frac{r_1}{2}+\frac{2k-4k^2}{2n}+\frac{i}{n}\frac{\Delta_r+4k-1}{2} 
= \frac{r_1}{2}+\frac{2k-4k^2+4ki}{2n}+\frac{i(\Delta_r -1)}{2n}
\end{align*}
To maximize the above as a function of $k$ it suffices to maximize $2k-4k^2+4ik$. Using the first order approach and equalizing the derivative $2-8k+4i$ to zero, we get that the maximum is at $k=\frac{i}{2}+\frac{1}{4}$. Since $k$ can only be an integer we take the closest integer to $\frac{i}{2}+\frac{1}{4}$. If $i$ is even, $\frac{i}{2}$ is the closest integer to the maximum. If $i$ is odd, the closest integer to the maximum is at $k=\frac{i+1}{2}$.

Then, we need to show that the agent's utility (as a function of $c$) coincides with the  breakpoints of the allocation rule as specified in Claim~\ref{cla:menu-allocation-rule}.
First, we show that the expected payment for action $i$ is $\frac{r_1}{2}+\frac{i\Delta_r}{2n}+\frac{i^2}{2n}$ as follows. If $i=2k$,
\begin{align*}
\frac{r_1}{2}+\frac{2k-4k^2}{2n}+\frac{i}{n}\frac{\Delta_r+4k-1}{2}=
\frac{r_1}{2}+\frac{i-i^2}{2n}+\frac{i}{n}\frac{\Delta_r+2i-1}{2}=
\frac{r_1}{2}+\frac{i\Delta_r}{2n}+\frac{i^2}{2n}.
\end{align*}
If $i=2k-1$,
\begin{align*}
& \frac{r_1}{2}+\frac{2k-4k^2}{2n}+\frac{i}{n}\frac{\Delta_r+4k-1}{2}\\
&= \frac{r_1}{2}+\frac{i+1-(i+1)^2}{2n}+\frac{i}{n}\frac{\Delta_r+2(i+1)-1}{2}\\
&= \frac{r_1}{2}+\frac{i-2i-i^2}{2n}+\frac{i}{n}\frac{\Delta_r+2i+1}{2}\\
&= \frac{r_1}{2}+\frac{-i-i^2}{2n}+\frac{i\Delta_r +2i^2+i}{2n}
= \frac{r_1}{2}+\frac{i\Delta_r}{2n}+\frac{i^2}{2n}.
\end{align*}
Then, the utility of the agent when taking action $i+1$ is higher than that of action~$i$ if and only if
\begin{eqnarray*}
\frac{r_1}{2}+\frac{(i+1)\Delta_r}{2n}+\frac{(i+1)^2}{2n}-\frac{(i+1)^2}{n} c &\geq& \frac{r_1}{2}+\frac{i\Delta_r}{2n}+\frac{i^2}{2n}-\frac{i^2}{n} c.
\end{eqnarray*}
That is, if and only if $\frac{\Delta_r}{2n}+\frac{2i+1}{2n} \geq \frac{2i+1}{n} c$, or equivalently, $\frac{\Delta_r}{4i+2}+\frac{1}{2} \geq c.$ This coincides with the allocation rule. Action $1$ is better than $0$ if and only if $\frac{r_1}{2}+\frac{1\Delta_r}{2n}+\frac{1^2}{2n}-\frac{1}{n}c\geq 0$. That is, $\frac{n r_1+\Delta_r+1}{2}\geq c$, again as required by the allocation rule.
\end{proof}

\begin{claim}
\label{cla:menu-LB}
The virtual welfare maximizing allocation rule is not implementable by a contract with menu-size complexity less than $\frac{n-1}{2}$.
\end{claim}

\begin{proof}
Suppose towards a contradiction that there exists a contract $(x,t)$ with menu-size complexity $<\frac{n-1}{2}$. Denote the payment schemes in $t$ as $t^1,t^2,...,t^k$ for $k<\frac{n-1}{2}$. Note that every for action chosen by the agent, the chosen payment scheme is the one guaranteeing the highest expected payoff for that action. There exists a payment scheme $t^{\ell}$ that gives the highest expected payment (among all payment schemes) to at least $3$ non-zero actions. Note that theses $3$ non-zero actions are successive. To see this, take $2$ non-successive action that $t^{\ell}$ gives the highest expected payment for, $0<i,i+j\in [n]$ for $j>1$. We show that $t^{\ell}$ gives the highest expected payment (among all contracts) for $i+1$ as well. We know that, 
\begin{eqnarray}
\frac{i}{n} t^{\ell}_{2}+\frac{n-i}{n} t^{\ell}_1 \geq \frac{i}{n} t^{\ell'}_2+\frac{n-i}{n} t^{\ell'}_1 & \forall \ell'\in [k] \label{eq:i}\\
\frac{i+j}{n} t^{\ell}_{2}+\frac{n-i-j}{n} t^{\ell}_{1} \geq \frac{i+j}{n} t^{\ell'}_2+\frac{n-i-j}{n} t^{\ell'}_1 & \forall \ell'\in [k] \label{eq:i_j}
\end{eqnarray}
If we multiply \eqref{eq:i} by $1-\frac{1}{j}$, and \eqref{eq:i_j} by $\frac{1}{j}$ and sum both equations we get.
\begin{eqnarray*}
(1-\frac{1}{j})(\frac{i}{n} t^{\ell}_{2}+\frac{n-i}{n} t^{\ell}_1)+\frac{1}{j}(\frac{i+j}{n} t^{\ell}_{2}+\frac{n-i-j}{n} t^{\ell}_{1}) &\geq& \\
(1-\frac{1}{j})(\frac{i}{n} t^{\ell'}_{2}+\frac{n-i}{n} t^{\ell'}_1)+\frac{1}{j}(\frac{i+j}{n} t^{\ell'}_{2}+\frac{n-i-j}{n} t^{\ell'}_{1}) && \forall \ell'\in [k].
\end{eqnarray*}
Reorganizing the above we have,
\begin{eqnarray*}
(\frac{i}{n} t^{\ell}_{2}+\frac{n-i}{n} t^{\ell}_1)+(\frac{1}{n} t^{\ell}_{2}-\frac{1}{n} t^{\ell}_{1}) \geq (\frac{i}{n} t^{\ell'}_{2}+\frac{n-i}{n} t^{\ell'}_1)+(\frac{1}{n} t^{\ell'}_{2}-\frac{1}{n} t^{\ell'}_{1}) &  \forall \ell'\in [k].
\end{eqnarray*}
Which implies that action contract $t^{\ell}$ gives the highest expected payment for action $i+1$ as follows.
\begin{eqnarray*}
(\frac{i+1}{n} t^{\ell}_{2}+\frac{n-i-1}{n} t^{\ell}_1) \geq (\frac{i+1}{n} t^{\ell'}_{2}+\frac{n-i-1}{n} t^{\ell'}_1) &  \forall \ell'\in [k].
\end{eqnarray*}
After establishing that $t^{\ell}$ incentivizes three successive actions $i,i+1,i+2$, we show that $x\neq x^*$ by showing that the breakpoints ratio of $x$ is not compatible with $x^*$'s breakpoint ratio. The first breakpoint (between action $i$ and action $i+1$ is at $\frac{t^{\ell}_2-t^{\ell}_1}{2i+1}$. To see this, equalize the utility of the agent from $t^{\ell}$ given action $i$ and action $i+1$ as follows. 
\begin{eqnarray*} 
\frac{i}{n} t_2 + (1-\frac{i}{n}) t_1 -\gamma_i c &=& \frac{i+1}{n} t_2 + (1-\frac{i+1}{n}) t_1 -\gamma_{i+1} c
\end{eqnarray*}
Reorganizing, and replacing $\gamma_i=\frac{i^2}{n}$ we have $(2i+1 )c = (t_2-t_1)$ at the intersection point.
We get the the ratio between the higher and the lower breakpoint ($\frac{t^{\ell}_2-t^{\ell}_1}{2i+1}$,$\frac{t^{\ell}_2-t^{\ell}_1}{2(i+1)+1}$) is $\frac{2(i+1)+1}{2i+1}$, whereas the allocation rule's breakpoints ratio is $\frac{\frac{\Delta_r }{4i+2}+\frac{1}{2}}{\frac{\Delta_r }{4(i+1)+2}+\frac{1}{2}}=\frac{2(i+1)+1}{2i+1}\frac{\Delta_r +{2i+1}}{\Delta_r +2(i+1)+1}\neq \frac{2(i+1)+1}{2i+1}$.
\end{proof}

\subsection{Monotone Virtual Welfare Maximization is not Sufficient}
\label{sub:failure_of_implementation}

Given Corollary~\ref{cor:upper-bound}, one may be tempted to optimize revenue by maximizing virtual welfare. Indeed one of the most appealing results of \citet{myerson1981optimal} is that optimizing revenue can be reduced to maximizing virtual welfare. In this section we show that virtual welfare maximization can be unimplementable when moral hazard is involved.\footnote{This result answers an open question in \cite{AlonDT21}.}
The proof relies on \cref{ex:appx-non-implement}, showing there exists no payment scheme for this allocation rule that satisfies the curvature constraints of Lemma~\ref{lem:primal characterization}.

\begin{proposition}
\label{pro:non-implement}
The virtual welfare maximizing allocation rule $x^*$ is not always implementable, even when monotone. 
\end{proposition}

\begin{example}\label{ex:appx-non-implement}
Consider the following setting with $4$ actions and $3$ outcomes. The costs in units-of-effort per action are $\gamma_0=0,\gamma_1=1,\gamma_2=3,\gamma_3=5.5$, and the rewards are $r_0=0,r_1=100,r_2=300$. The distributions over outcomes per action are $F_{0}=(1,0,0)$, $F_{1} =(0,1,0)$, $F_{2} =(0,0.5,0.5)$, $F_{3}=(0,0,1)$.
The type set is $C=[0,10]$. Taking $\delta = \frac{20}{23}$, the type distribution over $C$ is as follows.
\begin{eqnarray*}
g(c) = \begin{cases}
    \delta & c\in [0,1],\\
    0.025 \delta & c\in (1,4],\\
    0.0125 \delta & c\in (4,10],
    \end{cases} && G(c) = \begin{cases}
    \delta c & c\in [0,1],\\
    0.025  \delta (c-1) + \delta  & c\in (1,4],\\
    0.0125  \delta (c-4) + 1.075 \delta & c\in (4,10].
    \end{cases}
\end{eqnarray*}
\end{example}


\noindent Before proving Proposition~\ref{pro:non-implement} we establish the following claim.

\begin{claim}\label{clm:alloc_virtual_max_example}
Consider Example~\ref{ex:appx-non-implement}. The virtual welfare maximizing allocation rule is given by
\begin{eqnarray*}
x^*(c) = \arg\max_{i\in [n]}\{R_i-\gamma_i \varphi(c)\}= \begin{cases}
    3 & c\in [0,1),\\
    2 & c\in [1,4),\\
    1  & c\in [4,9),\\
    0 & c\in [9,10).
    \end{cases}
\end{eqnarray*}
\end{claim}

\begin{proof}
The virtual cost $\varphi(c)=c+\frac{G(c)}{g(c)}$ is as follows:
\begin{eqnarray*}
\varphi(c) = \begin{cases}
    c+c & c\in [0,1),\\
    c+c - 1 + \frac{1}{0.025} & c\in [1,4),\\
    c+c - 4 + \frac{ 1.075}{0.0125}  & c\in [4,10),
    \end{cases} && = \begin{cases}
    2c & c\in [0,1),\\
    2c + 39 & c\in [1,4),\\
    2c + 82  & c\in [4,10).
    \end{cases}
\end{eqnarray*}
The virtual welfare of every action is thus given by:
\begin{eqnarray*}
R_3-\gamma_3 \varphi(c) &=& \begin{cases}
    300-11c & c\in [0,1),\\
    -14.5 -11c & c\in [1,4),\\
    -151 -11c & c\in [4,10),
    \end{cases}   \\
    R_2-\gamma_2 \varphi(c)&=& \begin{cases}
    200-6c & c\in [0,1),\\
    83-6c & c\in [1,4),\\
    -46-6c & c\in [4,10),
    \end{cases} \\
 R_1-\gamma_1 \varphi(c) &=& \begin{cases}
    100-2c & c\in [0,1),\\
     61-2c & c\in [1,4),\\
     18-2c & c\in [4,10),
    \end{cases} 
\end{eqnarray*}
and $R_0-\gamma_0 \varphi(c)=0$. Therefore, the virtual welfare maximizing allocation rule is given by \begin{eqnarray*}
x^*(c) = \arg\max_{i\in [n]}\{R_i-\gamma_i \varphi(c)\}= \begin{cases}
    3 & c\in [0,1),\\
    2 & c\in [1,4),\\
    1  & c\in [4,9),\\
    0 & c\in [9,10),
    \end{cases}
\end{eqnarray*}
as claimed.
\end{proof}

\begin{proof}[Proof of \cref{pro:non-implement}]
The unique virtual welfare maximizing allocation rule $x^*$ in \cref{ex:appx-non-implement} assigns action $3$ for $c \in [0,1]$, action $2$ for $c \in (1,4]$, action $1$ for $c \in (4,9]$, and action $0$ for $c \in (9,10]$ (see \cref{clm:alloc_virtual_max_example}).
This allocation rule is monotone and can be verified using the definitions of virtual cost $\varphi(c) = c + G(c)/g(c)$ and virtual welfare $R_i - \gamma_i \varphi(c)$.
By \Cref{eq:curvature} of \cref{lem:primal characterization}, to establish that this rule is not implementable it suffices to show the following: For type $c=4$ and every payment scheme~$t^4\in \mathbb{R}^{m+1}$ that satisfies $x^*(4)=i^*(t^4,4)$, it holds that $\int_4^{1}\gamma_{x(z)}-\gamma_{i^*(t^4,z)}\dd z > 0$. I.e., any candidate payment scheme~$t^4$ must violate integral monotonicity, completing the proof.

To establish this, note that since $x^*(4)=i^*(t^4,4)$, action $x^*(4)=2$ gives higher utility for type $c=4$ than the utility from action~$1$. This implies that $0.5({t^4_1+t_2^4})-12 \geq t^4_1-4$. That is, $t^4_2-t^4_1 \geq 16.$ In this case, it can be verified that for types $c \leq c'$ where $c'$ is at least $3.2$ (the exact value of $c'$ depends on the precise payment scheme $t^4$), the utility from $t^4$ is maximized when taking action $3$, i.e, $i^*(t^4,c)=3$. 
Further, $i^*(t^4,c)=2$ for $[c',4]$. We thus have that $\int^{1}_{4} \gamma_{x^*(z)}-\gamma_{i^*(t^4,z)}\dd z$ is equal to
$\int^{c'}_1 \gamma_{3}-\gamma_{2}\dd z+ \int^{4}_{c'} \gamma_{2}-\gamma_{2}\dd z$, which is strictly positive. 
\end{proof}

We end this section by showing that the trouble in Example~\ref{ex:appx-non-implement} stems from the agent's ability to double-deviate --- to both an untruthful type and an unrecommended action.
In particular we show that types $c=1$ and $c=4-\epsilon\to 4$ cannot be jointly incentivized by any contract to take action~$2$. 
Suppose towards a contradiction that there exist such payment schemes $t^1=(t^1_0,t^1_1,t^1_2)$ and $t^4=(t_0^4,t_1^4,t_2^4)$ for types $c=1$ and $c=4-\epsilon$, respectively. By IC, the expected payment for action $2$ in both contracts is equal (otherwise both agents prefer one of these payment profiles). That is, $0.5(t^1_2+t^1_1)=0.5(t^4_2+t^4_1)$. Let us assume w.l.o.g.~that $t^1_2-t^1_1\geq t^4_2-t^4_1$. By IC, an agent of type $c=4$ is better off by reporting truthfully and taking action $2$, over taking action $1$. Therefore, $0.5({t^4_1+t_2^4})-12 \geq t^4_1-4$. That is, $t^4_2-t^4_1 \geq 16.$ Similarly, 
the agent of type $c=1$, is better off by reporting truthfully, and taking action $2$, over taking action $3$. This implies that $0.5(t^1_1+t^1_2)-3\geq t^1_2-5.5$. That is, $t^1_2-t^1_1\leq 5$. 
This implies that $5 \geq t^1_2-t^1_1\geq t^4_2-t^4_1 \geq 16$, a contradiction.

\end{document}